\documentclass{article}
\usepackage{kr}

\usepackage{times}
\usepackage{soul}
\usepackage{url}
\usepackage[hidelinks]{hyperref}
\usepackage[utf8]{inputenc}
\usepackage[small]{caption}
\usepackage{graphicx}
\usepackage{amsmath}
\usepackage{amsthm}
\usepackage{amssymb}
\usepackage{booktabs}
\usepackage[linesnumbered,ruled,vlined]{algorithm2e}
\usepackage{mathdots}
\usepackage{enumitem}
\usepackage{scalerel}

\usepackage{hyperref}\usepackage{thmtools} 
\usepackage{thm-restate}

\usepackage{tikz}
\usetikzlibrary{calc, shapes.arrows}

\newtheorem{theorem}{Theorem}
\newtheorem{lemma}[theorem]{Lemma}     
\newtheorem{corollary}[theorem]{Corollary}
\newtheorem{proposition}[theorem]{Proposition}
\newtheorem{claim}[theorem]{Claim}

\title{Two-Variable Logic for Hierarchically Partitioned and Ordered Data}

\author{%
Oskar Fiuk$^1$\and
Emanuel Kiero\'nski$^1$\and
Vincent Michielini$^2$
\affiliations
$^1$Institute of Computer Science, University of Wroc\l{}aw,\\
$^2$Faculty of Mathematics, Informatics, and Mechanics, Warsaw University\\
\emails
307023@uwr.edu.pl,
emanuel.kieronski@cs.uni.wroc.pl,
michielini@mimuw.edu.pl
}

\begin{document}

\maketitle

\hyphenation{e-xists}
\hyphenation{mis-sing}
\hyphenation{pa-ra-do-xi-cal}

\newcommand{\cD}{\mathcal{D}}
\newcommand{\cE}{\mathcal{E}}
\newcommand{\cP}{\mathcal{P}}
\newcommand{\cF}{\mathcal{F}}
\newcommand{\cQ}{\mathcal{Q}}
\newcommand{\cO}{\mathcal{O}}
\newcommand{\cI}{\mathcal{I}}
\newcommand{\cC}{\mathcal{C}}
\newcommand{\cR}{\mathcal{R}}
\newcommand{\cU}{\mathcal{U}}
\newcommand{\cS}{\mathcal{S}}
\newcommand{\cH}{\mathcal{H}}
\newcommand{\cN}{\mathcal{N}}
\newcommand{\cV}{\mathcal{V}}

\newcommand{\cB}{\mathcal{B}}
\newcommand{\cT}{\mathcal{T}}
\newcommand{\bC}{\mathbf{C}}
\newcommand{\cK}{\mathcal{K}}
\newcommand{\cG}{\mathcal{G}}
\newcommand{\cL}{\mathcal{L}}
\newcommand{\bbP}{\mathbb{P}}
\newcommand{\fA}{\mathfrak{A}}
\newcommand{\fB}{\mathfrak{B}}
\newcommand{\fC}{\mathfrak{C}}
\newcommand{\fD}{\mathfrak{D}}
\newcommand{\fG}{\mathfrak{G}}
\newcommand{\ff}{\mathfrak{f}}
\newcommand{\fg}{\mathfrak{g}}

\renewcommand{\phi}{\varphi} 
\newcommand{\eps}{\varepsilon} 

\newcommand{\AAA}{\mbox{\large \boldmath $\alpha$}}
\newcommand{\AAAp}{\mbox{\large \boldmath $\BBB'_{*}$}}
\newcommand{\BBB}{\mbox{\large \boldmath $\beta$}}

\newcommand{\Sat}{\ensuremath{\textit{Sat}}}
\newcommand{\FinSat}{\ensuremath{\textit{FinSat}}}

\newcommand{\FO}{\mbox{\rm FO}}
\newcommand{\FOt}{\mbox{$\mbox{\rm FO}^2$}}
\newcommand{\Ct}{\mbox{$\mathcal{C}^2$}}
\newcommand{\UOF}{\mbox{$\mbox{\rm UF}_1$}}
\newcommand{\ODF}{\mbox{$\mbox{\rm UF}_1$}}
\newcommand{\SUOF}{\mbox{$\mbox{\rm SUF}_1$}}
\newcommand{\RUOF}{\mbox{$\mbox{\rm RUF}_1$}}
\newcommand{\GFt}{\mbox{$\mbox{\rm GF}^2$}}
\newcommand{\GF}{\mbox{$\mbox{\rm GF}$}}
\newcommand{\GFTG}{\mbox{$\mbox{\rm GF+TG}$}}
\newcommand{\GFU}{\mbox{$\mbox{\rm GFU}$}}
\newcommand{\FF}{\mbox{$\mbox{\rm FF}$}}
\newcommand{\ALC}{$\cal ALC$}
\newcommand{\UF}{\mbox{$\mbox{\rm UF}_1$}}
\newcommand{\FAUF}{\mbox{$\forall\mbox{\rm -UF}$}}
\newcommand{\MK}{\mbox{$\mbox{\rm K}$}} 
\newcommand{\MDK}{\mbox{$\mbox{\rm DK}$}} 
\newcommand{\DMDK}{\mbox{$\overline{\mbox{\rm DK}}$}} 
\newcommand{\UNFO}{\mbox{$\mbox{\rm UNFO}$}}
\newcommand{\GNFO}{\mbox{$\mbox{\rm GNFO}$}}
\newcommand{\TGF}{\mbox{$\mbox{\rm TGF}$}}

\newcommand{\GC}{\mbox{$\mbox{\rm GC}$}}
\newcommand{\SC}{\mbox{$\mbox{\rm SC}$}}
\newcommand{\USC}{\mbox{$\mbox{\rm USC}$}}
\newcommand{\GSC}{\mbox{$\mbox{\rm GSC}$}}
\newcommand{\DMK}{\mbox{$\overline{\mbox{\rm K}}$}} 

\newcommand{\CGC}{\mbox{$\bigwedge\mbox{\rm GC}$}}
\newcommand{\CSC}{\mbox{$\bigwedge\mbox{\rm SC}$}}
\newcommand{\CUSC}{\mbox{$\bigwedge\mbox{\rm USC}$}}
\newcommand{\CGSC}{\mbox{$\bigwedge\mbox{\rm GSC}$}}
\newcommand{\CDMK}{\mbox{$\bigwedge{\mbox{\rm MC}}\overline{\mbox{\rm K}}$}} 

\newcommand{\NLogSpace}{\textsc{NLogSpace}}
\newcommand{\NP}{\textsc{NPTime}}
\newcommand{\PTime}{\textsc{PTime}}
\newcommand{\PSpace}{\textsc{PSpace}}
\newcommand{\ExpTime}{\textsc{ExpTime}}
\newcommand{\ExpSpace}{\textsc{ExpSpace}}
\newcommand{\NExpTime}{\textsc{NExpTime}}
\newcommand{\TwoExpTime}{2\textsc{-ExpTime}}
\newcommand{\TwoNExpTime}{2\textsc{-NExpTime}}
\newcommand{\ThreeNExpTime}{3\textsc{-NExpTime}}
\newcommand{\APSpace}{\textsc{APSpace}}
\newcommand{\TOWER}{\textsc{Tower}}

\newcommand{\str}[1]{{\mathfrak{#1}}}
\newcommand{\restr}{\!\!\restriction\!\!}

\newcommand{\N}{{\mathbb N}}   
\newcommand{\Q}{{\mathbb Q}}   
\newcommand{\Z}{{\mathbb Z}}   

\newcommand{\sss}{\scriptscriptstyle}

\newcommand{\tp}{{\rm tp}}
\newcommand{\type}[2]{{\rm tp}^{{#1}}[{#2}]}
\newcommand{\absclass}[3]{E^{{#1}}_{{#2}}[{#3}]}
\newcommand{\tet}[2]{{\mathfrak{t}({#1},{#2})}}

\newcommand{\ax}{{\rm Ax}}
\newcommand{\tr}{{\rm Tr}}

\newcommand{\rs}{\bar{r}}
\newcommand{\vs}{\bar{v}}
\newcommand{\xs}{\bar{x}}
\newcommand{\ys}{\bar{y}}
\newcommand{\zs}{\bar{z}}
\newcommand{\cs}{\bar{c}}

\newcommand{\grade}{\textrm{grade}}
\newcommand{\atoms}{\textrm{atoms}}
\newcommand{\freevars}{\textrm{freevar}}
\newcommand{\Cons}{{\rm{Cons}}}
\newcommand{\Rels}{{\rm{Rels}}}
\newcommand{\Vars}{\rm{Vars}}
\newcommand{\arity}{\textrm{ar}}
\newcommand{\mex}{\textrm{mex}}
\newcommand{\proj}{\textrm{proj}}
\newcommand{\rrightarrow}{\mathrel{\mathrlap{\rightarrow}\mkern1mu\rightarrow}}
\newcommand{\SAT}{\rm{SAT}}
\newcommand{\VER}{\rm{VER}}
\newcommand{\GAME}{\textrm{G}}
\newcommand{\POS}{\textrm{POS}}
\newcommand{\chc}{\mathfrak{chc}}
\newcommand{\pat}{\mathfrak{pat}}
\newcommand{\hull}{\textrm{hull}}

\newcommand{\bK}{K}
\newcommand{\bM}{M}

\newcommand{\eqdef}{:=}

\newcommand{\F}{\mathbb{F}}
\newcommand{\Field}{\textrm{GF}}

\newcommand{\poly}{\textrm{poly}}

\newcommand{\Perms}{\textrm{Perms}}
\newcommand{\tuple}[1]{\langle{#1}\rangle}

\newcommand{\arc}{\rightarrow}
\newcommand{\CColour}{\textrm{Control}_{\mu}}

\newcommand{\CRole}{\textrm{Control}_{\mu}}
\newcommand{\CBit}{\textrm{Control}_{\textrm{bit}}}
\newcommand{\Class}{\textrm{Class}}
\newcommand{\mask}{\textrm{repr}}
\newcommand{\select}{\textrm{select}}

\newcommand{\Cell}{\textrm{Cell}}

\newcommand{\bit}{\textrm{bit}}

\newcommand{\omegav}{\omega_{\textrm{V}}}
\newcommand{\omegas}{\omega_{\textrm{S}}}

\let\oldemptyset\emptyset
\let\emptyset\varnothing

\newcommand{\eqsiema}{\stackrel{\text{def}}{=\joinrel=}}

\newcommand{\extfun}{\mathfrak{ext}}

\newcommand{\neqv}{n_{eq}}
\newcommand{\nvars}{n_{vars}}
\newcommand{\nconj}{n_{conj}}

\newcommand{\RN}[1]{%
  \textup{\uppercase\expandafter{\romannumeral#1}}%
}

\newcommand{\num}{\mathfrak{n}}

\newcommand{\precEQ}{\ensuremath{{\preceq}^{\scaleto{\subseteq}{5pt}}}}
\newcommand{\succEQ}{\ensuremath{{\preceq}_{\rm{succ}}^{\scaleto{\subseteq}{5pt}}}}

\newcommand{\EQ}{\ensuremath{{\mathcal{EQ}}^{\scaleto{\subseteq}{5pt}}}}
\newcommand{\KEQ}[1]{\ensuremath{{{#1}\text{-}\mathcal{EQ}}^{\scaleto{\subseteq}{5pt}}}}

\begin{abstract}
  We study Two-Variable First-Order Logic, \FOt{}, under semantic constraints that model hierarchically structured data.
  Our first logic extends \FOt{} with a linear order $<$ and a chain of increasingly coarser equivalence relations $E_1 \subseteq E_2 \subseteq \ldots$.
  We show that its finite satisfiability problem is \NExpTime-complete.
  We also demonstrate that a weaker variant of this logic without the linear order enjoys the exponential model property.
  Our second logic extends \FOt{} with a chain of nested total preorders $\preceq_1 \subseteq \preceq_2 \subseteq \ldots$.
  We prove that its finite satisfiability problem is also \NExpTime-complete.
  However, we show that the complexity increases to \ExpSpace-complete once access to the successor relations of the preorders is allowed.
  Our last result is the undecidability of \FOt{} with two independent chains of nested equivalence relations.
\end{abstract}

\section{Introduction}\label{sec:introduction}

Hierarchically partitioned data are pervasive in modern computer systems.
For example, Geographical Information Services often organise geospatial information using progressively more detailed fields: country, region, state, and city.
Similar hierarchies appear in numerous contexts: Data Storage (organised into folders, subfolders, and files), Network Management (addressing schemes like IPv4/IPv6 with subnet hierarchies), Dependency Maintenance (tools for tracking dependencies between modules, libraries, and services).

To model such hierarchical data, we consider domains in which elements are annotated with data values drawn from potentially infinite or very large domains. A key aspect is that these data values can be tested for equality at multiple levels of precision. This is naturally captured by a family of increasingly coarser equivalence relations: two elements are related by the $k$-th equivalence relation if the $k$-th level equality test holds between their associated data values.

In this work, we establish results on the decidability and complexity of satisfiability problems for several variants of the Two-Variable Fragment of First-Order Logic, \FOt, extended to support such increasingly coarser equivalence relations. 
Our goal is to provide a logical framework that is expressive enough to model complex multi-level relationships while retaining desirable computational properties for reasoning. As most of the investigated logics do not enjoy the finite model property, their general and finite satisfiability problems differ. Our primary focus is on finite satisfiability, while the case of general satisfiability is left for future work.

The motivation for studying \FOt{} stems from its good algorithmic and model-theoretic properties. \FOt{} combines an \NExpTime-complete satisfiability problem and the exponential model property~\cite{GKV97} with a reasonable expressive power. In particular, it embeds (via the so-called \emph{standard translation}) many modal, temporal, and description logics (up to $\cal{ALCIOH^{\cap,\neq}}$).
Also, it is the maximal, in terms of the number of variables, fragment of First-Order Logic with decidable satisfiability problem, as already the Three-Variable Fragment is undecidable \cite{KMW62}. 
In the last few decades, \FOt{} together with its variations have been extensively studied, and plenty of results have been obtained in various scenarios (cf.~Subsection \ref{s:related} on related work). All of this makes \FOt{} often the first-choice option for various reasoning tasks.

In the following subsections of this introduction, we define our logics and present the obtained results (Subsection~\ref{s:contrib}), discuss related work (Subsection~\ref{s:related}), compare the expressive power of our logics (Subsection~\ref{s:applications}), and outline the technical sections that follow (Subsection~\ref{s:outline}).

\subsection{Logics of Interest and Our Results} \label{s:contrib}

Our underlying formalism is the Two-Variable Fragment of First-Order Logic, \FOt, whose formulas may use only the variables $x$ and $y$; any number of unary and binary \emph{common} relation symbols (i.e., with unconstrained interpretations); the equality symbol; and constant symbols, but no function symbols of positive arity. Extensions of \FOt{} are denoted by listing the \emph{special} symbols (i.e., those with constrained interpretations) in brackets, e.g., \FOt$[<,\EQ]$. This notation makes it explicit which additional semantic constraints are imposed on top of the base \FOt{} syntax.

\smallskip
\noindent
{\bf Order on domain elements.}
Let $\EQ$ denote the family of special symbols $E_1,E_2,\dots$ whose interpretations are constrained to \emph{nested} equivalence relations, that is for every $k \in \N$:
(i) the interpretation of $E_{k}$ is an equivalence relation, and
(ii) the interpretation of $E_{k+1}$ is \emph{coarser} than that of $E_{k}$ (i.e., $x E_k y \rightarrow x E_{k+1} y$).
Let $<$ be a special symbol interpreted as a \emph{strict} linear order on domain elements.\footnote{Throughout the paper, we will also use the derived non-strict linear order $\le$, and their analogues for other orders, e.g., $\prec$ and $\preceq$. Notice that in \FOt{} $\le$ is definable from $<$ and vice-versa.}

Our first considered logic is $\FOt[<,\EQ]$, supporting both a linear order on elements and hierarchical data values.

Potential applications of $\FOt[<,\EQ]$ can be found in, e.g., temporal verification of multiprocess systems.
We provide a motivational example of enforcing isolation policy in an environment of processes, containers, and events.

System execution is modeled as a linear sequence of events ordered by time using the relation $<$.
Events are generated by processes, and processes are grouped into containers.
The equivalence relations $E_1$ and $E_2$ capture this hierarchy: $x E_1 y$ holds when events $x$ and $y$ come from the same process, and $x E_2 y$ when they come from possibly distinct processes yet running in the same container.
These are naturally nested: $x E_1 y \rightarrow x E_2 y$.

Importantly, the domain of our model consists only of events: processes are implicitly represented as equivalence classes of the relation $E_1$.
That is, all events belonging to the same process form a single equivalence class of $E_1$. (We assume that every process reports at least one event, e.g., \emph{spawn event}, to ensure that it has a non-empty class.) Likewise, containers are implicitly represented as equivalence classes of the relation $E_2$. (We assume that in every container at least one process is running, e.g., \emph{root process}.)

Process-level properties (e.g., \emph{sandboxed}, \emph{privileged}) are expressed via class-wise unary predicates. For example:
\begingroup
\setlength{\abovedisplayskip}{0.4em}
\setlength{\belowdisplayskip}{0.4em}
\[
  \forall x, y.~x E_1 y \rightarrow \big(\mathrm{sandboxed}(x) \leftrightarrow \mathrm{sandboxed}(y)\big)
\]
\endgroup
This ensures that all events from the same process share the same \emph{sandboxed} or \emph{not sandboxed} label, even though the logic only quantifies over events.

Using the above described encoding, we can impose the following isolation policy:
\emph{``A sandboxed process must not communicate with events outside its container unless an explicit grant was made before from a privileged process.''}

This policy can be formalised with a sentence:
\begin{align*}
  &\forall x_1.~\big(\mathrm{sandboxed}(x_1) \land \phi_{\mathrm{cross\text{-}container\text{-}message}}(x_1)\big) \rightarrow \\
  &~~~~\big(\exists y_1.~y_1<x_1 \land x_1{E_1}y_1 \wedge \phi_{\mathrm{permission\text{-}grant}}(y_1) \big),
\end{align*}
where $\phi_{\mathrm{cross\text{-}container\text{-}message}}(x_1)$ stands for the formula
\begingroup
\setlength{\abovedisplayskip}{0.4em}
\setlength{\belowdisplayskip}{0.4em}
\[\exists y_2.~x_1 < y_2 \land \neg x_1 E_2 y_2 \land \mathrm{message}(x_1,y_2),\]
\endgroup
and $\phi_{\mathrm{permission\text{-}grant}}(y_1)$ for
\begingroup
\setlength{\abovedisplayskip}{0.4em}
\setlength{\belowdisplayskip}{0.4em}
\[\exists x_3.~\mathrm{privileged}(x_3) \land x_3 < y_1 \land \mathrm{grant}(x_3,y_1).\]
\endgroup

In the above formulas, we annotated variables as $x_1$, $y_1$, $y_2$, $x_3$ for readability to reflect their roles in different processes.
However, only two variables $x$ and $y$, with reuse across quantifiers, are sufficient to express this property.

As stated above, all variables refer to events: $x_1$ is a \emph{dispatch message} event from Process~1 (the \emph{sender}) that initiates the communication, and $y_2$ is a \emph{deliver message} event from Process~2 (the \emph{receiver}) that receives the message from Process~1, as indicated by $\mathrm{message}(x_1, y_2)$. The event $y_1$ is an earlier event from the same process as $x_1$, representing the \emph{grant acknowledge} event, marking the point at which Process~1 becomes aware of and is authorised to act on the granted permission. Finally, $x_3$ is a \emph{grant authorise} event from Process~3 (a \emph{privileged admin}) that issues the permission and notifies Process~1, as indicated by $\mathrm{grant}(x_3, y_1)$.

The intended temporal order of events $x_3{<}y_1{<}x_1{<}y_2$ is enforced:
the permission is issued before it is acknowledged ($x_3{<}y_1$);
the message is sent after the permission being acknowledged ($y_1{<}x_1$) and before being delivered ($x_1{<}y_2$).

Notice also that, to detect that Process~1 and Process~2 are running in distinct containers, we refer to their representative events using $\neg x E_2 y$. Since $E_1$ is nested within $E_2$, it follows that if two events are not $E_2$-related, then their respective $E_1$-classes (i.e., processes) must also belong to distinct $E_2$-classes (i.e., containers).

\smallskip
Our first main contribution is establishing the complexity of the finite satisfiability problem for \FOt{}$[<,\EQ]$:

\begin{restatable}{theorem}{secondTheorem}\label{t:two}
  Finitely satisfiable sentences of \FOt$[<,\EQ]$ admit models of exponential size.
  The finite satisfiability problem for \FOt$[<,\EQ]$ is \NExpTime-complete.
\end{restatable}

Notice that $\FOt{}[<,\EQ]$ does not include the induced successor of $<$.\footnote{Notice that the induced successor is definable in $\FO$, yet not in $\FOt{}$, by the formula $\phi(x,y) \eqdef x < y \wedge \forall z.~(x < z \rightarrow y \le z)$.}
\cite{BB07} show that adding it leads to undecidability (see: Subsection~\ref{s:related}).

\smallskip
\noindent
{\bf Order on data values.}
We consider now a different way to incorporate a linear order into our scenario. We trade an order on domain elements for a family of nested linear orders on data values, i.e., on equivalence classes of $E_1,E_2,\dots$.

Let $\precEQ$ denote the family of special symbols $\preceq_1,\preceq_2,\dots$ whose interpretations are constrained to \emph{nested} total preorders, that is for each $k \in \N$: (i) the interpretation of $\preceq_{k}$ is a total preorder\footnote{A \emph{total preorder} is a transitive relation $\preceq$ such that, for every $x,y$, either $x \preceq y$ or $y \preceq x$ holds (in particular, $x \preceq x$ holds).}, and (ii) the interpretation of $\preceq_{k}$ is a subrelation of the interpretation of $\preceq_{k+1}$ (i.e. $x{\preceq_k}y \rightarrow x{\preceq_{k+1}}y$).

For example, interpretations over $\N$ defined by $n \preceq_k m$ iff $\lfloor \frac{n}{10^k} \rfloor \le \lfloor \frac{m}{10^k} \rfloor$ satisfy the above requirements.

Our second logic is $\FOt[\precEQ]$, supporting hierarchical data values that can be compared at multiple levels of granularity using less-than, equal, and greater-than comparison tests.

In \FOt{}$[\precEQ]$, we naturally keep nested equivalence relations $E_1,E_2,\dots$ Their interpretation is now given by the $\preceq_k$-equivalent elements: $x E_k y \leftrightarrow x \preceq_k y \wedge y \preceq_k x$.

Natural applications of $\FOt[\precEQ]$ arise in temporal reasoning tasks that involve increasingly fine-grained notions of time.
A representative case study is a system supporting atomic transactions.
Here, the elements of the structure represent low-level operations.
The finest preorder $\preceq_1$ models the underlying timeline, where $\preceq_1$-equivalent elements are treated as occurring in parallel.
Transactions are modeled implicitly as $E_2$-equivalence classes, ordered chronologically by $\preceq_2$.
Coarser preorders can be used to represent higher-level structures such as $\preceq_3$ for commit order, $\preceq_4$ for versioning, and so on.
With this encoding, we can express that, e.g., ``\emph{raised exceptions must be handled in the future, yet within the same transaction}'' as follows:
\begin{align*}
    \forall x.~\mathrm{exception}(x) \rightarrow \exists y.~x \prec_1 y \land x E_2 y \land \mathrm{handles}(y, x)
\end{align*}

For \FOt$[\precEQ]$, we establish the following theorem:

\begin{theorem}\label{t:three}
  Finitely satisfiable sentences of \FOt$[\precEQ]$ admit models of exponential size.
  The finite satisfiability problem for \FOt$[\precEQ]$ is \NExpTime-complete.
\end{theorem}

\smallskip
\noindent
{\bf Adding successors on data values.}
We consider now a logic \(\FOt[\succEQ]\) that enriches the syntax of \(\FOt[\precEQ]\) by adding, for each \(\preceq_k\)-symbol, its induced successor predicate \(\cS_k\), defined as
\(
  \cS_k(x,y) \eqdef x{\prec_k}y \land \forall z.~\big(x{\prec_k}z \rightarrow y{\preceq_k}z\big).
\)

Applications of \(\FOt[\succEQ]\) naturally extend those of \(\FOt[\precEQ]\); we continue with the transaction-based system scenario.
For example, we can express the property:  
\emph{``If a transaction fails, then in the immediate \(\preceq_1\)-successor time-slot a rollback must occur, using the last available snapshot from the immediate \(\preceq_2\)-predecessor transaction.''} Formally:
\begin{align*}
  &\forall x.~\mathrm{fail}(x) \rightarrow 
    \big(\exists y.~\mathrm{rollback}(y) \land \cS_1(x,y) \land x E_2 y \ \land \\
  &\quad\exists x.~\mathrm{last\text{-}snapshot}(x) \land \cS_2(x,y) \land \mathrm{restore\text{-}to}(y,x)\big)
\end{align*}
Here, \(\cS_1(x,y) \land x E_2 y\) ensures that the rollback follows the failure in the very next \(\preceq_1\)-time-slot within the same transaction,  
while \(\cS_2(x,y)\) ensures that the snapshot belongs to the transaction immediately preceding the one containing the rollback.  
In particular, this guarantees that the snapshot was created before the failure.
Additional axioms can be imposed to refine the scenario further: for instance, we can define \(\mathrm{last\text{-}snapshot}\) as the \(\preceq_1\)-maximum snapshot within each \(E_2\)-class,  
or require that after a rollback the system proceeds directly to the \(\preceq_2\)-successor transaction.

Now we state the second main contribution of this paper:

\begin{restatable}{theorem}{fourthTheorem}\label{t:four}
  Finitely satisfiable sentences of \FOt$[\succEQ]$ admit models of doubly exponential size.
  The finite satisfiability problem for \FOt$[\succEQ]$ is \ExpSpace-complete.
\end{restatable}

\smallskip\noindent
{\bf Absence of orders.}
Both \FOt$[<,\EQ]$ and \FOt$[\precEQ]$ admit sentences that enforce infinite models (e.g., $\forall x.~\exists y.~x{<}y$ and $\forall x.~\exists y.~x{\prec_1}y$).
This contrasts with pure \FOt{} which enjoys the \emph{finite model property} (i.e., every satisfiable sentence has a finite model).
A natural question arises: Can nested equivalence relations alone enforce infinite models, or is this phenomenon solely due to the presence of linear orders?

Let \FOt$[\EQ]$ denote \FOt{} extended with the family of nested equivalence relations $E_1,E_2,\dots$.
We answer that \FOt$[\EQ]$ does indeed enjoy the finite model property:

\begin{theorem}\label{t:one}
  Satisfiable sentences of \FOt$[\EQ]$ admit models of exponential size.
  The satisfiability and finite satisfiability problems for \FOt$[\EQ]$ coincide and are \NExpTime-complete.
\end{theorem}

\smallskip\noindent
{\bf Undecidability.}
A natural candidate to explore next is a logic supporting two independent families of nested equivalence relations $E_1,E_2,\dots$ and $F_1,F_2,\dots$.
We prove that reasoning in such a logic is undecidable---even in a very restricted setting---when each family has length~$2$ and the vocabulary, except the four special equivalence symbols, is composed of unary predicates only.

\begin{restatable}{theorem}{undecTheorem}\label{t:und}
  The satisfiability and finite satisfiability problems are undecidable for the constant-free, equality-free, monadic fragment of \FOt{}
  extended with four special symbols $E_1, E_2, F_1, F_2$,
  interpreted as equivalence relations such that $E_2$ is coarser than $E_1$ and $F_2$ is coarser than $F_1$. 
\end{restatable}

\subsection{Related Work} \label{s:related}

\smallskip\noindent
\textbf{\FOt{} with equivalence relations.} 
\FOt{} with a single equivalence relation has the exponential model property and an \NExpTime-complete satisfiability problem~\cite{KO12}. 
With two equivalence relations (not necessarily nested), the finite model property is lost; both satisfiability and finite satisfiability are \TwoNExpTime-complete~\cite{KMP-HT14}. 
With three equivalence relations, \FOt{} is undecidable~\cite{KO12}.

Since \FOt{} can express containment between equivalence relations, the work of~\cite{KMP-HT14} established decidability for \FOt{} over hierarchical partitions of depth two (with no ordering). 
However, it does not yield optimal complexity bounds even in this restricted setting and leave the question of the finite model property unanswered.

\smallskip\noindent
\textbf{\FOt{} with linear orders.} 
The satisfiability and finite satisfiability problems for $\FOt$ with a single linear order are \NExpTime-complete~\cite{Ott01}. 
With three linear orders, both problems become undecidable~\cite{Kie11}. 
The case of two linear orders was studied in~\cite{HZ16}, where finite satisfiability is shown to be decidable in \TwoNExpTime\ when one linear order is accessible via both the order and successor predicates, and the other only via one of them. Decidability of the general satisfiability for \FOt{} with two linear orders follows from \cite{TZ20} (up to our knowledge, the complexity of this problem has not been studied.)

\smallskip\noindent
\textbf{Data Words.} 
\FOt{} with a linear order and an equivalence relation has been studied extensively in the context of \emph{data words}. 
These are structures interpreting unary predicates, an equivalence relation, a linear order, and its induced successor relation. 
The satisfiability problem for \FOt{} on data words is decidable but non-elementary-hard~\cite{BDM11}. 
When the successor relation is omitted, the problem becomes \NExpTime-complete.

Quite close to our setting are data words with nested equivalences, that where considered in~\cite{BB07}.
It is shown there that satisfiability is decidable when the linear order is accessible only via its successor relation,
and becomes undecidable as soon as both the order and its successor are accessible.
The variant closest to \FOt$[<,\EQ]$, where only the linear order (and not its successor) is available, was not explored in that work. 
Notice that data words restrict the common part of the signature to unary symbols, whereas we allow full \FOt{} with arbitrary binary relations that may freely interact with the equivalences. 

\smallskip\noindent
\textbf{\FOt{} with total preorders.} 
The \ExpSpace-completeness of the finite satisfiability problem for \FOt{} over structures with one total preorder, its induced successor relation, a linear order, and additional unary relations was established in~\cite{SZ12}. 
As in the case of data words, no common binary relations are permitted. 
When two independent total preorders are available, satisfiability becomes undecidable. 
To the best of our knowledge, the case of a single total preorder combined with arbitrary binary relations, as well as settings with nested total preorders, have not been investigated so far.

\smallskip\noindent
\textbf{\FOt{} over trees.}
An alternative perspective on nested equivalence relations is to interpret them as trees.
If only $k$ nested equivalence symbols $E_1,\dots,E_k$ are considered, structures can be viewed as modeling the leaves of unranked trees of fixed depth $k$, where $E_i(a,b)$ holds if $a$ and $b$ share a common ancestor at depth $k{-}i$.
Importantly, only leaf nodes constitute the domain; internal tree nodes serve an auxiliary role and are not part of the universe.
This interpretation is different from standard two-variable logics over trees, cf.~\cite{BMS09,CW16b,BBC16}, where all tree nodes belong to the domain and the structure is accessed via navigational predicates such as \emph{parent}, \emph{child}, \emph{descendant}, etc.

\smallskip\noindent
\textbf{Other logics.}
Nested equivalence relations can be simulated in description logics such as $\mathcal{SHI}$ and $\mathcal{SHOI}$, which include transitive roles ($\mathcal{S}$), inverse roles ($\mathcal{I}$), role hierarchies ($\mathcal{H}$), and, possibly, nominals ($\mathcal{O}$). 
These logics have \ExpTime-complete satisfiability problems. 
However, interactions between binary relations (roles) are limited to role hierarchies, meaning one can only express containment between relations (these can be equivalences or common relations). Moreover, these logics do not include linear orderings.

Among first-order fragments, the Unary Negation Fragment, \UNFO{}~\cite{StC13}, is particularly worth to mention.
This logic restricts negation to subformulas with at most one free variable. 
Its extension capturing $\mathcal{SHOI}$, denoted \UNFO+$\mathcal{SOH}$, is decidable and \TwoExpTime-complete~\cite{DK19}.
Remarkably, \UNFO{}+$\mathcal{SOH}$ enables reasoning over arbitrarily many independent families of nested equivalence relations.

\subsection{Expressivity of \FOt$[<,\EQ]$ and \FOt$[\precEQ]$} \label{s:applications}

Naturally, $\FOt[\succEQ]$ is more expressive than $\FOt[\precEQ]$, yet both formalisms are incomparable with $\FOt[<,\EQ]$.
In $\FOt[\succEQ]$ and $\FOt[\precEQ]$, every $E_k$-class necessarily forms an interval with respect to a linear order induced from the total preorder $\preceq_1$ by resolving ties arbitrarily.
In contrast, this interval property cannot be expressed in $\FOt[<,\EQ]$. 
On the other hand, $\FOt[<,\EQ]$ can enforce that certain classes do not form intervals:
\begingroup
\setlength{\abovedisplayskip}{0.4em}
\setlength{\belowdisplayskip}{0.4em}
\begin{align*}
	&\forall x.~\exists y.~ x<y \wedge \neg x E_1 y \wedge \big(P(x) \leftrightarrow \neg P(y)\big)\\
	&\forall x,y.~\big(P(x)\wedge P(y)\big)\to x E_1 y
\end{align*}
\endgroup

\subsection{Outline of Technical Sections} \label{s:outline}

Section~\ref{sec:preliminaries} introduces the necessary notions and definitions.  
Section~\ref{s:A} proves Theorem~\ref{t:two}.
Section~\ref{s:sketch} sketches the proofs of Theorems~\ref{t:three} and~\ref{t:one}, as these are similar to that of Theorem~\ref{t:two}; more details are in  Appendix~\ref{s:BC}.
Theorem~\ref{t:four} is established in Section~\ref{s:D}, with certain technical details deferred to Appendix~\ref{appendix:D}
and Appendix~\ref{appendix:lower}.
Section~\ref{sec:undecidability} presents the undecidability result of Theorem~\ref{t:und}.
Finally, Section~\ref{s:conclusion} concludes the paper.

\section{Preliminaries}\label{sec:preliminaries}
\subsection{Notation and Conventions}

We denote the set of natural numbers including $0$ by $\N$. For $k \in \N$, the notation $[k]$ stands for the set $\{1, \dots, k\}$, with the convention that $[0] = \emptyset$. More generally, we use interval notation $[a, b] \subseteq \N$ to denote the set $\{a, a+1, \dots, b\}$ whenever $a \leq b$, and the empty set $\emptyset$ whenever $a > b$.
If $E$ is an equivalence relation on a set $A$, then $\absclass{}{}{a}$ denotes the equivalence class of an element $a \in A$, and $B/E$ denotes the quotient set via $E$ of a subset $B \subseteq A$.

A \emph{signature} $\sigma$ is a finite set of symbols, partitioned as $\sigma = \Cons \cup \Rels$, where $\Cons$ is the set of constant symbols and $\Rels$ is the set of relation symbols (including special symbols such as $<$, $E_1$, $E_2$, etc.). Every relation symbol has associated arity. We do not allow function symbols of positive arity. The \emph{signature of a formula} is the finite set of relation and constant symbols that appear in the formula.

The \emph{size} (or \emph{length}) of a formula $\phi$, denoted $|\phi|$, is defined as the total number of symbols it contains, where each occurrence of a symbol---be it a variable, relation symbol, or constant---contributes $1$ to the count.

We use Fraktur letters such as $\str{A}, \str{B}, \dots$ to denote structures, and the corresponding Roman letters $A, B, \dots$ for their domains.
A structure $\str{A}$ over a signature $\sigma$ interprets the symbols in $\sigma$: a relation symbol $R$ as a relation $R^{\str{A}}\subseteq A^k$ with $k$ denoting the arity of $R$; and a constant symbol $c$ as an element $c^{\str{A}} \in A$.
If $B \subseteq A$, we write $\str{A} \restr B$ for the \emph{restriction} of $\str{A}$ to the subdomain $B$. The \emph{size} of a structure is the cardinality of its domain.
Elements of structures are typically denoted by $a, b, \dots$; variables by $x, y$, possibly with decorations. 
We write $\phi(\xs)$ to indicate that all free variables of the formula $\phi$ are contained in the tuple $\xs$.

An (atomic) \emph{$1$-type} over a signature $\sigma$ is a maximal consistent set of literals involving only the variable $x$.
Similarly, an (atomic) \emph{$2$-type} over $\sigma$ is a maximal consistent set of literals over the variables $x$ and $y$; in particular, a $2$-type naturally determines the $1$-types of both $x$ and $y$.
We use the symbol $\alpha$ (possibly with decorations) to denote $1$-types.

Let $\str{A}$ be a structure. For any $a \in A$, we write $\type{\str{A}}{a}$ to denote the unique $1$-type \emph{realised} by $a$ in $\str{A}$,
that is the $1$-type $\alpha$ such that $\str{A} \models \alpha(a)$.
Similarly, for distinct elements $a, b \in A$, the notation $\type{\str{A}}{a,b}$ denotes the unique $2$-type \emph{realised} in $\str{A}$ by the pair $(a,b)$.

\subsection{Scott Normal Form}
\label{s:snf}

A sentence $\phi$ is in \emph{Scott Normal Form}~\cite{Sco62} if it is in the shape of $\forall x, y.~\psi_0(x, y) \wedge \bigwedge_{m=1}^{M}{\forall x.~\exists y.~\psi_m(x, y)}$, where $M$ denotes the number of Skolem conjuncts.
For $m \in [0, M]$, the formula $\psi_m(x, y)$ is quantifier-free and constant-free, and uses only relation symbols of arity $1$ and $2$.

There is a standard polynomial-time procedure transforming any $\FOt$ sentence into Scott Normal Form, preserving satisfiability over the same domains; see e.g.,~\cite{PH23}.
It is readily verified that it is sound for the extensions of \FOt{} introduced in Subsection~\ref{s:contrib}.

We standardise the special relation symbols in $\phi$ by assuming they are the consecutive symbols $E_1, \dots, E_K$ for some $K \in \N$.
We do the same in the case of $\preceq_k$-symbols.

For convenience, we interpret additional symbol $E_0$ as the identity relation, and 
assume that $E_{K+1}$, the first symbol not used in $\phi$, is interpreted as the universal relation. Note that this preserves nestedness: $E_0\subseteq E_1$ and $E_K\subseteq E_{K+1}$.

\section{Proof of Theorem~\ref{t:two}}\label{s:A}

In this section, we establish our first main result.

\secondTheorem*

For this whole section, we fix a sentence $\phi$ from the logic $\FOt[<,\EQ]$.
Assuming that $\phi$ is finitely satisfiable, we aim to prove that it is actually satisfied by a model of exponential size.
As discussed in Subsection~\ref{s:snf}, we may assume that $\phi$ is in Scott Normal Form, i.e.,
in the shape of \(\forall x,y.~\psi_0(x,y) \wedge \bigwedge_{m=1}^{M}{\forall x.~\exists y.~\psi_m(x,y)} \),
where $M$ denotes the number of Skolem conjuncts.
By convention, $\phi$ uses the first $K$ nested equivalence symbols $E_1,\dots,E_K$; in addtion, we interpret $E_0$ and $E_{K+1}$ as the identity and universal relations, respectively.
Let $\AAA$ denote the set of all $1$-types over the signature of $\phi$.

The main ingredient of our proof is the following lemma, which allows us to replace a single equivalence class with its small counterpart, without touching the rest of the structure.

\begin{lemma}\label{l:replacement}
Let $\str{A}$ be a finite model of $\phi$.
Fix $k \in [0,K]$, and let $\cC$ be an $E_{k+1}$-class in $A/E^{\str{A}}_{k+1}$.
Then there exists a subset $\cD\subseteq\cC$ such that:
\begin{enumerate}
	\item $\cD$ is the union of at most $12 \cdot M^3 \cdot |\AAA|$ many $E_k$-classes from $\cC/E_k^{\str{A}}$.
  \item $\phi$ has a model $\str{B}$ over the domain $B = (A\setminus\cC)\cup\cD$.
\end{enumerate}
\end{lemma}

We will prove Lemma~\ref{l:replacement} in a moment. Let us first demonstrate how to derive Theorem~\ref{t:two} from it.

\begin{proof}[Proof of Theorem~\ref{t:two}]
  Assume that $\varphi$ is finitely satisfiable, and consider a finite model $\str{A}$ of minimal size.
  
  Because of this minimality, for each $k\in[0,K]$, every $E_{k+1}$-class of $\str{A}$ is partitioned by at most $12\cdot M^3\cdot|\AAA|$ many $E_k$-classes. As otherwise, Lemma~\ref{l:replacement} would give us a model of strictly smaller size. Since each $E_0$-class of $\str{A}$ is of size $1$, an immediate induction tells us that, for each $k \in [K{+}1]$, the $E_{k}$-classes of $\str{A}$ are of size at most $12^k \cdot M^{3k} \cdot |\AAA|^k$.
  
  Yet, $E^{\str{A}}_{K{+}1}$ having only one equivalence class, the size of $\str{A}$ is therefore bounded by $12^{K+1} \cdot M^{3K+3} \cdot |\AAA|^{K+1}$. Clearly, $K$ and $M$ are both $\cO(|\phi|)$.
  As $\phi$ does not use constants, the number of $1$-types in $\AAA$ is $2^{|\Rels|} = 2^{\cO(|\phi|)}$.
  
  We conclude that the size of $\str{A}$ is indeed $2^{\cO(|\phi|^2)}$.
\end{proof}

We now prove Lemma~\ref{l:replacement}.

Suppose that $\str{A}$ is a finite model of $\phi$, $k$ is a number between $0$ and $K$, and $\cC$ is any $E_{k+1}$-class of $A/E^{\str{A}}_{k+1}$.

We assume w.l.o.g.~that the domain $A$ is a finite subset of~$\N$ and that the order $<^{\str{A}}$ of $\str{A}$ coincides with the usual order on $\N$, allowing us to write just $<$ instead of $<^{\str{A}}$.

We use the following terminology:
if $m \in [M]$ and $a,b$ are elements of $A$, then we call $b$ a \emph{witness} of $a$ whenever $\str{A} \models \psi_m(a,b)$.
In addition, $b$ is an \emph{internal} witness if $(a,b) \in E^{\str{A}}_k$; an \emph{external} witness if $b \in A \setminus \cC$ and $(a,b) \not\in E^{\str{A}}_k$; and a \emph{$\cC$-witness} if $b \in \cC$ and $(a,b) \not \in E^{\str{A}}_k$.

We fix \emph{witness functions} $\ff_m\colon A \rightarrow A$, i.e., such that for each $m \in [M]$ and every $a \in A$, $\str{A}\models \psi_m(a,\ff_m(a))$.

\smallskip
Before delving into the details, let us outline our proof strategy.
We construct the set $\cD \subseteq \cC$, to replace the selected $E_{k+1}$-class $\cC$, in three steps:
First, we define a subset $W_{\RN{1}} \subseteq \cC$ that will serve as witnesses for elements outside $\cC$.  
Next, we set $W_{\RN{2}} \eqdef \bigcup_{m \in [M]} \ff_m[W_{\RN{1}}]$,
which will provide witnesses for elements in $W_{\RN{1}}$.  
Similarly, we define  
$W_{\RN{3}} \eqdef \bigcup_{m \in [M]} \ff_m[W_{\RN{2}}]$,  
to serve as witnesses for elements in $W_{\RN{2}}$.
It will be enough to define $\cD$ as the $E_k$-closure of $(W_{\RN{1}} \cup W_{\RN{2}} \cup W_{\RN{3}}) \cap \cC$. Finally, the structure $\str{B}$ will be defined as $\str{A}~\restr~[(A\setminus\cC)\cup\cD]$ but with certain $2$-types redefined.

\smallskip
In order to replace the original $\cC$-witnesses by other elements, we use the following notion of \emph{configurations}.

Let $\ell \in [M]$ and $a \in A$.
Assume that $\langle b_1, \dots, b_\ell \rangle$ and $\langle c_1, \dots, c_\ell \rangle$ are tuples of pairwise distinct elements of $A$.
We say that they realise the same \emph{$a$-configuration} if the following conditions hold:
\begin{enumerate}[label = (\roman*)]
  \item\label{configurationitems1} For every $i \in [\ell]$, $\type{\str{A}}{b_i} = \type{\str{A}}{c_i}$.
  \item\label{configurationitems2} For every $i \in [\ell]$, $a \bowtie b_i$ if and only if $a \bowtie c_i$, where $\bowtie$ is any of the symbols ``$<$'', ``$=$'' or ``$>$''.
  \item\label{configurationitems3} For every $i \in [\ell]$ and every $j \in [K]$, $(a,b_i)\in E^{\str{A}}_{j}$ if and only if $(a,c_i)\in E^{\str{A}}_{j}$.
\end{enumerate}

We define now the set $W_{\RN{1}} \subseteq \cC$.
For this, we introduce the notion of the \emph{$r$-extremal} subset of a finite subset $S\subseteq\N$, of $\ell$ elements $a_1< a_2<\ldots< a_{\ell}$: if $\ell\leq2\cdot r$, then this subset, denoted by ${\rm extr}_r(S)$, is $S$ itself; if $\ell> 2\cdot r$, then it is $\{ a_1,\dots,a_{r},a_{\ell-r+1},\dots,a_\ell \}\subset S$.

For each $1$-type $\alpha$, we take the set $V_\alpha$ of the $M$ smallest and $M$ largest realisations of $\alpha$ from every $E_k$-class in $\cC$: \[V_\alpha\eqdef\bigcup_{\cE \in \cC/E^{\str{A}}_k}{{\rm extr}_M(\{ a \in \cE \mid \type{\str{A}}{a} = \alpha \})}.\]

We then define $W_{\RN{1}}$ as the set of $2  M$ smallest and $2  M$ largest elements of each $V_\alpha$: \[W_{\RN{1}}\eqdef\bigcup_{\alpha \in \AAA}{{\rm extr}_{2 \cdot M}(V_\alpha)}.\]

In the following claim, we justify that any configuration of $\cC$-witnesses can be realised within $W_{\RN{1}}$.
\begin{claim}\label{l:configurations}
  Let $\ell \in [M]$, and let $a \in A \setminus W_{\RN{1}}$. Suppose that $b_1,\dots,b_\ell$ are distinct elements of $\cC \setminus \absclass{\str{A}}{k}{a}$.
  Then there exist distinct elements $c_1,\dots,c_\ell$ in $W_{\RN{1}}$ such that $\langle b_1,\dots,b_\ell \rangle$ and $\langle c_1,\dots,c_\ell \rangle$ realise the same $a$-configuration.
\end{claim}
\begin{proof}
  We can partition the $b_i$'s into several subtuples by the $1$-types that these elements realise.
  Considering these tuples independently, we obtain several tuples of the $c_i$'s to be glued together into a final tuple.
  Hence, w.l.o.g.~assume that all the $b_i$'s realise the same $1$-type $\alpha$, for some $\alpha \in \AAA$.

  Similarly, we can moreover assume that the $b_i$'s are related to $a$ in the same way, say $b_i < a$ holds for every $i \in [\ell]$.

  Consider the set $S \eqdef \{ c \in \cC \mid c < a \text{ and } \type{\str{A}}{c}=\alpha \}$.
  Notice that $\{ b_1,\dots,b_\ell \} \subseteq S \setminus \absclass{\str{A}}{k}{a}$.
  Moreover, any $\ell$-tuple of distinct element from $S \setminus \absclass{\str{A}}{k}{a}$ realises the same $a$-configuration as the $b_i$'s.
  Hence, select $\ell$ minimal elements $c_1,\dots,c_\ell$ from the set $S \setminus \absclass{\str{A}}{k}{a}$.
  Because $\ell \le M$, we have that $\{ c_1,\dots,c_\ell\} \subseteq V_\alpha$.
  Now, crucially among the $2\cdot M$ minimal elements of the set $V_\alpha$, in addition to the elements $c_1,\dots,c_\ell$, there are at most $M$ additional ``\emph{bad}'' elements from the set $S \cap \absclass{\str{A}}{k}{a}$.
  In consequence, we have that $\{ c_1,\dots,c_\ell \} \subseteq {\rm extr}_{2\cdot M}(V_\alpha)$, which is a subset of $W_{\RN{1}}$.
\end{proof}

Naturally, elements of $W_{\RN{1}}$ require witnesses themselves, which is why, as announced above, we define  
$W_{\RN{2}} \eqdef \bigcup_{m \in [M]}\ff_m[W_{\RN{1}}]$.
Similarly, elements of $W_{\RN{2}}$ need witnesses, so we define  
$W_{\RN{3}} \eqdef \bigcup_{m \in [M]}\ff_m[W_{\RN{2}}]$.

We do not introduce further sets $W_{\RN{4}}, W_{\RN{5}},\dots$, as elements of $W_{\RN{3}} \setminus (W_{\RN{1}} \cup W_{\RN{2}})$ can be assigned $\cC$-witnesses again in $W_{\RN{1}}$.  
This step resembles the construction of the exponential model property for \FOt{} by Gr\"adel, Kolaitis, and Vardi~\cite{GKV97}, where three sets of elements provide witnesses for each other in a circular fashion.  
Our setting is slightly more intricate, but the general idea remains: $W_{\RN{2}}$ provides $\cC$-witnesses for $W_{\RN{1}}$, $W_{\RN{3}}$ for $W_{\RN{2}}$, and, to close the cycle, we adjust the $2$-types so that $W_{\RN{1}}$ provides $\cC$-witnesses for $W_{\RN{3}} \setminus (W_{\RN{1}} \cup W_{\RN{2}})$.  
Moreover, $W_{\RN{1}}$ needs to provide $\cC$-witnesses for elements outside $W_{\RN{1}} \cup W_{\RN{2}} \cup W_{\RN{3}}$.  
We ensure all of this in the subsequent steps of the proof.

We define the set $\cD$ (from the statement of Lemma~\ref{l:replacement}) as:
\[ \cD:=\bigcup_{a \in (W_{\RN{1}} \cup W_{\RN{2}} \cup W_{\RN{3}}) \cap \cC} \absclass{\str{A}}{k}{a}. \]

This choice of $\cD$ satisfies the first item of Lemma~\ref{l:replacement}:
\begin{claim}\label{l:sizeanalysis}
	The set $\cD$ is a union of at most $12 \cdot M^3 \cdot |\AAA|$ many $E_k$-classes of $\cC/E^{\str{A}}_k$.
\end{claim}
\begin{proof}
  We have that $|W_{\RN{1}}| \le 4 \cdot M \cdot |\AAA|$, $|W_{\RN{2}}| \le M \cdot |W_{\RN{1}}|$, and $|W_{\RN{3}}| \le M \cdot |W_{\RN{2}}|$.
  Therefore, each of the sets $W_{\RN{1}}$, $W_{\RN{2}}$ and $W_{\RN{3}}$ is of size at most $4 \cdot M^3 \cdot |\AAA|$.
  We conclude, as $|\cD| \le |W_{\RN{1}} \cup W_{\RN{2}} \cup W_{\RN{3}}| \le 12 \cdot M^3 \cdot |\AAA|$.
\end{proof}

To conclude this proof, it remains to establish also the second item of Lemma~\ref{l:replacement}:
the sentence $\phi$ has a model $\str{B}$ over the domain $B \eqdef (A \setminus \cC) \cup \cD$.
In this structure, the $1$-types, the linear order $<$, and the relations $E_k$ are induced from $\str{A}$, yet some $2$-types will be redefined in a specific way.

In particular, this redefinition will ensure the elements of $B \setminus (W_{\RN{1}} \cup W_{\RN{2}})$ to have $\cC$-witnesses in $W_{\RN{1}}$.

Let $N$ denote $|B \setminus (W_{\RN{1}} \cup W_{\RN{2}})|$.
We fix an enumeration $a_1,\dots,a_N$ of the set $B \setminus (W_{\RN{1}} \cup W_{\RN{2}})$.
We inductively construct a sequence of structures $\str{B}_0,\str{B}_1,\dots,\str{B}_N$: $\str{B}_0$ is the induced structure $\str{A} \restr B$, and, for each $n \in [0,N{-}1]$, $\str{B}_{n+1}$ is obtained from $\str{B}_n$ by modifying some $2$-types between $a_{n+1}$ and $W_{\RN{1}} \setminus \absclass{\str{A}}{k}{a_{n+1}}$ in the way described below.

Let $W_\cC(a_{n+1})$ denote the set of $\cC$-witnesses of $a_{n+1}$ suggested by the witness functions $\ff_m$:
\[ W_\cC(a_{n+1})\eqdef\{ \ff_m(a_{n+1}) \mid m \in [M] \} \cap (\cC \setminus \absclass{\str{A}}{k}{a_{n+1}}). \]

We set $b_1,\dots,b_\ell$ as the distinct elements of $W_\cC(a_{n+1})$, with $\ell \leq M$.
From Claim~\ref{l:configurations}, there are distinct elements $c_1,\dots,c_\ell$ from $W_{\RN{1}}$ realising the same $a_{n+1}$-configuration as $b_1,\dots,b_\ell$.
We define the structure $\str{B}_{n+1}$ almost as $\str{B}_{n}$, except that for each $c_j$ (which is in $W_{\RN{1}}\subseteq B$), the $2$-type between $a_{n+1}$ and $c_j$ is defined as the $2$-type between $a_{n+1}$ and $b_j$, i.e., $\tp^{\str{B}_{n+1}}[a_{n+1},c_j] \eqdef \tp^{\str{B}_n}[a_{n+1},b_j]$.

The following claim summarises the properties of $\str{B}_{n+1}$.

\begin{claim}\label{l:smallclaimfo2}
  Suppose that the numbers $n$ and $N$, the elements $a_1,\dots,a_N$ of $B \setminus (W_{\RN{1}} \cup W_{\RN{2}})$, and the structure $\str{B}_{n+1}$ are all defined as above.
  Then:
  \begin{enumerate}\itemsep0pt
    \item\label{smallclaimfo2item1} $\str{B}_{n+1}$ realises only the $1$-types and $2$-types from $\str{A}$.
    \item\label{smallclaimfo2item2} The linear order $<^{\str{B}_{n+1}}$ and the equivalence relations $E^{\str{B}_{n+1}}_j$, for every $j \in [K]$, are all inherited from $\str{A}~\restr~B$.
    \item\label{smallclaimfo2item3} For all distinct $a,b \in B$, $\type{\str{A}}{a,b} = \type{\str{B}_{n+1}}{a,b}$ unless $(a,b)$ or $(b,a)$ is in $(\{ a_1,\dots,a_{n+1} \} \times W_{\RN{1}}) \setminus E^{\str{A}}_k$.
    \item\label{smallclaimfo2item4} For each $i \in [n{+}1]$ and each $b \in W_\cC(a_i)$, there exists $c \in W_{\RN{1}}$ such that $\type{\str{A}}{a_i,b} = \type{\str{B}_{n+1}}{a_i,c}$.
  \end{enumerate}
\end{claim}

From the items of Claim~\ref{l:smallclaimfo2}, we conclude that the structure $\str{B}$, defined as the final structure $\str{B}_{N}$, is a model of $\varphi$ over the domain $(A \setminus \cC) \cup \cD$.
This finishes the proof of the second item of Lemma~\ref{l:replacement}.

\section{How to show Theorems \ref{t:three} and \ref{t:one}} \label{s:sketch}

Theorems~\ref{t:three} and~\ref{t:one} follow by similar arguments as Theorem~\ref{t:two}.
We sketch them here; details are in Appendix~\ref{s:BC}.

\smallskip
\noindent
\emph{Theorem~\ref{t:three}.}
We proceed via a reduction to \FOt$[<,\EQ]$.
Let $\str{A}_0$ be a finite $\FOt[\precEQ]$-model of $\varphi$.
We expand it to an $\FOt[<,\EQ]$-structure $\str{A}$ by:
(i) setting $E_k^{\str{A}} \eqdef {\preceq}_k^{\str{A}_0} \cap ({\preceq}_k^{\str{A}_0})^{-1}$ for all $k$;
and (ii) interpreting $<^{\str{A}}$ as a linear order derived from ${\preceq_1}^{\str{A}_0}$ by resolving ties within $E_1$-classes arbitrarily.
Applying Lemma~\ref{l:replacement} to $\str{A}$ yields a model $\str{B}$ over the domain $B = (A \setminus \cC) \cup \cD$ with $\cD \subseteq \cC$.
Item~\ref{smallclaimfo2item2} of Claim~\ref{l:smallclaimfo2} implies that $\preceq^{\str{B}}_k = \preceq_k^{\str{A} \ \restr\  B}$ for all $k$.
Since nested preorders are preserved under taking substructures, $\str{B}$ respects the semantics of $\preceq_k$-symbols, and thus it can be naturally transformed back into an $\FOt[\precEQ]$-model of $\varphi$.
Since the size bound analysis from the proof of Theorem~\ref{t:two} still applies, we get a small model and thereby establish Theorem~\ref{t:three}.

\smallskip
\noindent
\emph{Theorem~\ref{t:one}.}
The argument proceeds analogously to that for Lemma~\ref{l:replacement} and Theorem~\ref{t:two}, but now starting from a possibly infinite model $\str{A}$.
In the absence of the linear order, we may choose arbitrary realisations of $1$-types into $r$-extremal subsets (rather than the maximal and minimal ones).
The rest of the reasoning carries over directly, with the only modification being that, due to the possible infinity of $\str{A}$, it is sometimes necessary to work with natural limit structures.

\section{Proof of Theorem~\ref{t:four}}\label{s:D} 

In this section, we establish our second main result.

\fourthTheorem*

For this section, we fix a sentence $\phi$ of $\FOt[\succEQ]$, and assume that it is in Scott Normal Form, i.e., in the shape of \(\forall x,y.~\psi_0(x,y) \wedge \bigwedge_{m=1}^{M}{\forall x.~\exists y.~\psi_m(x,y)} \), where $M$ denotes the number of Skolem conjuncts.
As explained in Subsection~\ref{s:snf}, the problem of transforming any $\FOt$ sentence into Scott Normal Form, satisfiable over the same domains, is in polynomial-time. Hence, this assumption is without loss of generality.
By convention, $\phi$ uses the first $K$ special symbols $\preceq_k$ for $k \in [K]$, and the corresponding derived symbols $E_k$ and $\cS_k$; moreover we interpret the symbols $E_0$ and $E_{K+1}$ as the identity and universal relations, respectively.

\smallskip
\noindent{\bf Small model property.}
We begin with the first part of Theorem~\ref{t:four}, establishing that every finitely satisfiable sentence of $\FOt[\succEQ]$ has a model of doubly exponential size.  
The key ingredient is the following lemma, analogous to Lemma~\ref{l:replacement}; its proof is deferred to Appendix \ref{appendix:D}.

\begin{restatable}{lemma}{pumpinglemma}\label{l:pumping}
Let $\str{A}$ be a finite model of $\phi$.
Fix $k \in [0,K]$, and let $\cC$ be an $E_{k+1}$-class in $A / E^{\str{A}}_{k+1}$.
Then there exists a subset $\cD \subseteq \cC$ such that:
\begin{enumerate}\itemsep0pt
  \item $\cD$ is the union of $2^{2^{\cO(|\phi|)}}$ many $E_k$-classes from $\cC / E^{\str{A}}_k$.
  \item $\phi$ has a model $\str{B}$ with domain $B = (A \setminus \cC) \cup \cD$.
\end{enumerate}
\end{restatable}

The argument for deriving the small model bound from Lemma~\ref{l:pumping} is parallel to that of Theorem~\ref{t:two} from Lemma~\ref{l:replacement}.

Let $\str{A}$ be a minimal-size finite model.
By Lemma~\ref{l:pumping}, for each $k \in [0,K]$, every $E_{k+1}$-class of $\str{A}$ is partitioned into at most $2^{2^{\cO(|\phi|)}}$ many $E_k$-classes.
Since each $E_0$-class is a singleton and $E_{K+1}^{\str{A}}$ has a single equivalence class, it follows that $\str{A}$ is indeed of doubly exponential size:
\[
|A| \leq \Big(2^{2^{\cO(|\phi|)}}\Big)^{K+1} = 2^{(K+1) \times 2^{\cO(|\phi|)}} = 2^{2^{\cO(|\phi|)}}.
\]

\smallskip
\noindent
{\bf Satisfiability-checking procedure.}
We now turn to the second part of Theorem~\ref{t:four}, establishing an \ExpSpace{} procedure deciding the finite satisfiability problem for \FOt$[\succEQ]$.
A preliminary version of this procedure, operating in 2-\ExpSpace{}, is given in Algorithm~\ref{algos}.

\SetKw{KwGuess}{guess}
\SetKw{KwReject}{reject}
\SetKw{KwAccept}{accept}
\SetCommentSty{emph}        
\begin{algorithm}
  \caption{Finite Satisfiability for \FOt$[\succEQ]$}\label{algos}
  \KwIn{A sentence $\phi$ in Scott Normal Form}
  \KwOut{Decides whether $\phi$ is finitely satisfiable}
  \KwGuess the domain size $N \in \N$\;
  \For{$i \leftarrow 1,2, \dots, N$}{
    \KwGuess the number $r_i \in [K{+}1]$\;
    \KwGuess the $1$-type $\alpha_i \in \AAA$\;
    \lIf{$\alpha_i \not\models \psi_0(x,x)$}{\KwReject}
    $W_i \leftarrow \{ m \in [M] \mid \alpha_{i}\models \psi_m(x,x) \}$\;
    \For{$j \in \{ 1, \dots, i{-}1 \}$}{
      \If{$q_j\le r_i$}{$p_j \leftarrow q_j$\; $q_j \leftarrow r_i$\;}
      \KwGuess the $2$-type $\beta_{j,i} \in \BBB[\alpha_j,\alpha_i,p_j,q_j]$\;
      \lIf{$\beta_{j,i} \not\models\psi_0(x,y) \land \psi_0(y,x)$}{\KwReject}
      $W_j \leftarrow W_j \cup \{ m \in [M] \mid \beta_{j,i}\models \psi_m(x,y) \}$\;
      $W_i \leftarrow W_i \cup \{ m \in [M] \mid \beta_{j,i}\models \psi_m(y,x) \}$\;
    }
    $p_i \leftarrow 1$; $q_i \leftarrow 1$\;
  }
  \lIf{$W_i = [M]$ for every $i \in [N]$}{\KwAccept}
  \lElse{\KwReject}
\end{algorithm}

We call a $1$-type or $2$-type \emph{proper} if it interprets all the special symbols $\preceq_k$, $E_k$, and $\cS_k$ in a manner consistent with the semantics of \FOt$[\succEQ]$.  
In other words, proper types are exactly those that can occur in actual models of \FOt$[\succEQ]$.  
In particular, they respect the natural compatibility conditions between equivalences and preorders, e.g., if $E_k(x,y)$ holds, then $x \preceq_k y$ and $y \preceq_k x$ must both hold as well.  
For $2$-types, we additionally require that they entail $x \preceq_k y$ for the relevant~$k$; this condition guarantees compatibility with the fixed ordering alignment assumed in the algorithm.
Denote by $\AAA$ and $\BBB$ the sets of proper, respectively, $1$-types and $2$-types over the signature of $\phi$.

In the procedure, we use a table $\BBB[*,*,*,*]$.
In this table, the entry $\BBB[\alpha_1,\alpha_2,p,q]$, for every $\alpha_1,\alpha_2 \in \AAA$ and $1 \le p \le q \le K{+}1$,
stores a set of $2$-types. We put a $2$-type $\beta \in \BBB$ into the set $\BBB[\alpha_1,\alpha_2,p,q]$ if the following conditions hold:
\begin{enumerate}[label = (\roman*)]
  \item $\type{\beta}{x} = \alpha_1$ and $\type{\beta}{y} = \alpha_2$.
  \item For each $k \in [K]$, $E_k(x,y) \in \beta$ iff $q \le k$.
  \item For each $k \in [K]$, $S_k(x,y) \in \beta$ iff $p \le k < q$.
\end{enumerate}

Algorithm~\ref{algos} attempts to verify the existence of a model $\str{A}$ for~$\phi$ over the domain $\{1,\dots,N\}$, for some $N \in \N$.
We assume that the preorders $\preceq^{\str{A}}_k$ are aligned with the natural order $\le$ on~$\N$ (i.e., $i{\prec_k}j \rightarrow i{<}j$).
In the $i$th iteration of the outer loop, the algorithm examines how the new element~$i$ interacts with all previous elements $1,\dots,i{-}1$, denoted by $j$ in the inner loop.
The variable $q_j$ stores the smallest index $k$ such that $j$ and $i$ belong to the same $E_k$-class; and $K{+}1$ is a sentinel value.  
If $p_j {<} q_j$, then $p_j$ records the smallest index $k$ for which $j$ and $i$ lie in $\preceq_k$-consecutive $E_k$-classes; if instead $p_j {=} q_j$, then no such consecutive $E_k$-classes exist.

To illustrate better the idea of Algorithm~\ref{algos}, we visualise the structure in the following schematic, each horizontal brace groups together the elements of one $E_k$-class, and $\preceq_k$ linearly orders the $E_k$-classes from left to right:
{\fontsize{7.9}{8}\selectfont
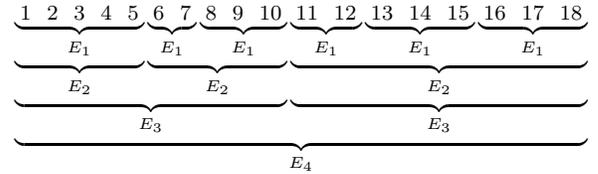
\begin{align*}
  \underbrace{
    \underbrace{
      \underbrace{
        \underbrace{~1~~~2~~~3~~~4~~~5~}_{E_1}
      }_{E_2}
      \underbrace{
        \underbrace{~6~~~7~}_{E_1}
        \underbrace{~8~~~9~~~10~}_{E_1}
      }_{E_2}
    }_{E_3}
    \underbrace{
      \underbrace{
        \underbrace{~11~~~12~}_{E_1}
        \underbrace{~13~~~14~~~15~}_{E_1}
        \underbrace{~16~~~17~~~18~}_{E_1}
      }_{E_2}
    }_{E_3}
  }_{E_4}
\end{align*}}

Consider now the $18$th iteration of the outer loop, that is, $i = 18$.
The values of the variables $p_j$ and $q_j$ for selected~$j$ are:
\(p_{5} = 3,\ q_{5} = 4;\ \)
\(p_{9} = 2,\ q_{9} = 4;\ \) 
\(p_{12} = 2,\ q_{12} = 2\).

Intuitively, these pairs $(p_j,q_j)$ capture the \emph{relative position} of~$j$ with respect to~$i$ in the hierarchy of equivalence classes, allowing the algorithm to determine which $2$-types between elements $j$ and $i$ are admissible.

The algorithm maintains, for each element $i \in [N]$, a set $W_i \subseteq [M]$ recording which Skolem conjuncts are already satisfied for~$i$, and accepts precisely when every element has all its witnesses secured, i.e., when $W_i = [M]$ for all $i$.

We formalise the above explanations in the following claim, establishing completeness of the procedure.
\begin{claim}
  Suppose that $\phi$ has a finite model $\str{A}$.
  Then there exists a sequence of guesses of Algorithm~\ref{algos} such that, after $i$ iterations, where $i \in [N]$, the following invariants hold:
\begin{itemize}
  \item For $j \in [i]$, $\type{\str{A}}{j}{=}\alpha_j$, and for $j{<}k{\le}i$, $\type{\str{A}}{j,k}{=}\beta_{j,k}$.
  \item If $i>1$, then, for $k \in [K]$, $(i{-}1,i) \in E^{\str{A}}_k$ iff $r_i \le k$. 
  \item For $j \in [i]$ and $k \in [K]$, $(j,i) \in E^{\str{A}}_k$ iff $q_j \le k$.
  \item For $j \in [i]$ and $k \in [K]$, $(j,i) \in \cS^{\str{A}}_{k}$ iff $p_j \le k < q_j$.
  \item For $j \in [i]$, $W_j$ indicates witnesses of $j$ located in $\str{A}~\restr~[i]$.
\end{itemize}
\end{claim}

Soundness follows as well.  
Given a transcription of an accepting run, we reconstruct a finite model of~$\phi$ as follows.  
Take the domain to be $\{1,\dots,N\}$, assign to each element~$i$ the $1$-type~$\alpha_i$, and for each ordered pair $(j,i)$ assign the $2$-type~$\beta_{j,i}$.
Since each pair is processed exactly once, and all $2$-types are determined consistently with both the preorder relations and their successors, no conflicts can arise.  
The resulting structure therefore satisfies~$\phi$.
\begin{corollary}
  Algorithm~\ref{algos} is sound and complete.
\end{corollary}

Algorithm~\ref{algos} operates in space proportional to~$N$.
By the already established small model property, $N$ can be assumed doubly exponential in~$|\phi|$, yielding a $2$-\ExpSpace{} bound.

However, we note that the precise identities of the previously processed elements $j < i$ are irrelevant: all information needed about such an element~$j$ is captured by the tuple $\langle \alpha_j, W_j, p_j, q_j \rangle$.  
Consequently, after the $i$th iteration, the algorithm's configuration can be represented compactly as a \emph{counting function}, recording the number of such tuples:
\[
F_i \colon \AAA \times 2^{[M]} \times [K{+}1] \times [K{+}1] \to \N.
\]
It is straightforward to adapt the procedure to work directly with this representation.  
Storing $F_i$ explicitly as a table requires  
\(
\cO\!\left(|\AAA| \cdot 2^M \cdot K^2 \cdot \log N\right)
\)  
bits of space.  
Since $|\AAA|$ is $2^{\cO(|\phi|)}$, $M$ and $K$ are $\cO(|\phi|)$, and $N$ is doubly exponential in~$|\phi|$, the overall space usage is singly exponential in~$|\phi|$.

\begin{corollary}
  Algorithm~\ref{algos} can be implemented to operate in exponential space.
\end{corollary}

\smallskip
\noindent
{\bf Lower bound.}
The final part of the proof of Theorem~\ref{t:four}
is to show the matching \ExpSpace{}-hardness lower bound for the finite satisfiability problem for $\FOt[\succEQ]$.

Here we give a short proof sketch; full details are deferred to Appendix \ref{appendix:lower}.
We reduce from the corridor tiling problem, which is \ExpSpace{}-complete.
An instance consists of a set of colours \(\cC\), horizontal and vertical adjacency constraints, colours for the top-left and bottom-right cells (initial conditions), and a number \(n\) in unary.
The decision question is whether there exists a number \(m\) and a labelling with colours from $\cC$ of a \(2^n \times m\) grid that satisfies all constraints.

The reduction is realised by a conjunction of \FOt[\succEQ] sentences that enforce the structure to encode such a grid.
Introduce unary predicates \(B_j\) for \(0\le j<n\) together with unary predicates \(P_c\) for each colour \(c\in\cC\).
Interpretation: each domain element is a grid cell; the predicates \(B_j\) encode the horizontal coordinate in binary (yielding width \(2^n\)), and each \(E_1\)-equivalence class represents a full row.
Within a row, \(P_c\) marks the colour of the cell at the horizontal position given by the bits.
The successor relation \(\cS_1\) links consecutive \(E_1\)-classes, so rows form a vertical chain of length \(m\).
Expressing the horizontal and vertical adjacency constraints, as well as the initial conditions, in \FOt[\succEQ] is routine.

\begin{corollary}
  The finite satisfiability problem for \FOt$[\succEQ]$ is $\ExpSpace$-hard.
\end{corollary}

\section{Proof of Theorem~\ref{t:und}}\label{sec:undecidability}

In this section, we establish our undecidability result.

\undecTheorem*

To the end of this section, we prove Theorem~\ref{t:und}.

The idea demonstrated here is similar to Pratt-Hartmann's proof of the undecidability of $\FOt$ with counting and two equivalence relations~\cite{PH14}.
Hovewer, in the absence of the equality symbol, additional technical refinements are required.

The undecidability is proven via a reduction from the halting problem for \emph{two-counter machines}. Such a machine $\mathbf{M}$ consists of a finite set $S=\{s_0,\ldots, s_{F}\}$ of states, one being initial (say $s_0$) and one being final (say $s_F$); two counters $c_1$ and $c_2$ holding non-negative integers; and a set of transitions $\delta\subseteq S\times Op\times S$ between states and operations over the counters; these operations are: increment a counter, decrement a counter (if it is non-zero), or test if a counter is zero. A \emph{configuration} is a triple $\langle s,c_1,c_2 \rangle$ holding the state $s$ and the values of the two counters $c_1$ and $c_2$. A \emph{run} is a (potentially infinite) sequence of configurations which starts with the initial state $s_0$ and respects the transition function. The machine \emph{terminates} if it ever reaches the final state $s_F$. As two-counter machines can simulate Turing machines, the halting problem for two-counter machines is undecidable (more precisely, it is r.e.-complete).

We encode a run of some two-counter machine $\mathbf{M}$.
Each configuration will be contained within an equivalence class of the intersection relation $E_2 \cap F_2$. We write $G_2(x,y)$ for the formula $E_2(x,y)\wedge F_2(x,y)$. We partition the domain with two unary predicates $d_E$ and $d_F$, and axiomatise that they are universal in each $G_2$-class:
	\begin{align}
		&\forall x.~d_E(x) \leftrightarrow \neg d_F(x)\\	
		&\forall x,y.~G_2(x,y)\rightarrow \big(d_E(x)\leftrightarrow d_E(y)\big)
	\end{align}

    If a configuration has its elements satisfying $d_E$ (resp. $d_F$), we call it a \emph{$d_E$-configuration} (resp. \emph{$d_F$-configuration}). We specify that each $E_2$- and $F_2$-class has at most one configuration of each type:
	\begin{align}
      &\bigwedge_{\alpha \in \{E,F\}} \forall x,y.~\big(E_2(x,y){\wedge}d_\alpha(x){\wedge}d_\alpha(y)\big){\rightarrow} F_2(x,y)\\
      &\bigwedge_{\alpha \in \{E,F\}} \forall x,y.~\big(F_2(x,y){\wedge}d_\alpha(x){\wedge}d_\alpha(y)\big){\rightarrow} E_2(x,y)
	\end{align}
	
    We define now that each configuration is in precisely one state $s$ from the set of states $S$ of $\mathbf{M}$:
	\begin{align}
		&\forall x.~\bigvee_{s\in S}\forall y.~G_2(x,y)\rightarrow \big(s(y)\wedge\bigwedge_{s' \in S\colon s' \neq s}\neg s'(y)\big)
	\end{align}

\begin{figure}
	\begin{center}
		\begin{tikzpicture}[scale = 0.4]
			
			
			\draw (0,0) -- (0,5) -- (5,5) -- (5,0) -- (0,0);
			
			\draw (-3,-1) -- (-3,4) -- (2,4) -- (2,-1) -- (-3,-1);
			
			\draw (3,-1) -- (3,4) -- (8,4) -- (8,-1) -- (3,-1);
			
			\draw (6,0) -- (6,5) -- (11,5) -- (11,0) -- (6,0);
			
			\draw (-6,0) -- (-6,5) -- (-1,5) -- (-1,0) -- (-6,0);
			
			\node at (-0.5, -1.5) {$F_2$};
			\node at (5.5, -1.5) {$F_2$};
			
			\node at (-3.5, 5.5) {$E_2$};
			\node at (2.5, 5.5) {$E_2$};
			\node at (8.5, 5.5) {$E_2$};
			
			\node at (-2, 3) {$d_F$};
			\node at (1, 3) {$d_E$};
			\node at (4, 3) {$d_F$};
			\node at (7, 3) {$d_E$};
			
			 \node[single arrow, draw=black, fill=white, 
			 minimum width = 10pt, single arrow head extend=3pt,
			 minimum height=10mm] at (-0.5, 1.5) {};
			 \node at (-0.5, 2.2) {$t$};
			 
			 \node[single arrow, draw=black, fill=white, 
			 minimum width = 10pt, single arrow head extend=3pt,
			 minimum height=10mm] at (2.5, 1.5) {};
			 \node at (2.5, 2.2) {$t$};
			 
			 \node[single arrow, draw=black, fill=white, 
			 minimum width = 10pt, single arrow head extend=3pt,
			 minimum height=10mm] at (5.5, 1.5) {};
			 \node at (5.5, 2.2) {$t$};
			
		\end{tikzpicture}
		\caption{A succession of configurations.}
		\label{fig:t_between_configs}
	\end{center}
\end{figure}
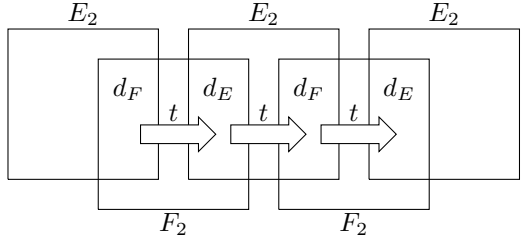

	We order configurations with a definable binary relation $t$: among an $E_2$-class (resp. $F_2$-class), it links elements from the $d_E$-configuration to the successive $d_F$-one (resp. from the $d_F$-one to the successive $d_E$-one), see Figure~\ref{fig:t_between_configs}:
	\begin{align*}
      &t(x,y) \eqdef \big(E_2(x,y)\wedge d_E(x)\wedge d_F(y)\big) \\
      &~~~~~~~~~~~~~\nonumber \vee\big(F_2(x,y)\wedge d_F(x)\wedge d_E(y)\big)
	\end{align*}

    We have discussed so far how to axiomatise the succession of configurations, we explain now how to implement the counters with the help of the two thinner equivalence relations $E_1$ and $F_1$.
    We write $G_1(x,y)$ for the formula $E_1(x,y) \vee F_1(x,y)$, defining the union of $E_1$ and $F_1$; a priori not an equivalence relation.
    We require first that inside any $G_2$-class, $E_1$ and $F_1$ define the same relation:
    \begin{align}
      \forall x,y.~\big(G_2(x,y){\wedge}G_1(x,y)\big) \rightarrow \big(E_1(x,y){\wedge}F_1(x,y)\big)
	\end{align}

    We introduce two unary predicates $c_1$ and $c_2$. Their intended meaning is as follows:
    If $\cE$ is a $G_2$-class (representing a configuration), then, inside it, the value of the counter $c_i$ is the number of equivalence classes of the relation $E_1 \cap \cE^2$ whose elements satisfy $c_i$. In particular, if no element of $\cE$ satisfies $c_i$, then the value of the counter $c_i$ equals zero.
	We permit elements to satisfy only a state predicate, without being marked by $c_1$ or $c_2$.  
	This ensures that configurations where both counters are zero can still be represented.

    We require first that $c_1$ and $c_2$ are mutually exclusive:
    \begin{align}
      \forall x.~\neg c_1(x) \vee \neg c_2(x)
    \end{align}
    We also state that $c_1$ and $c_2$ propagate within $G_1$-connected components (and hence they do not interact via $E_1$ or $F_1$):
	\begin{align}
      \bigwedge_{i=1}^{2}\forall x,y.~\big(G_1(x,y) \wedge c_i(x)\big) \rightarrow c_i(y)
	\end{align}

    Then, the change of the counters is held as follows.
    In a $d_E$-configuration, $E_1$ may extend to the successive $d_F$-configuration and $F_1$ may extend to the preceding $d_F$-configuration.
    And, in a $d_F$-configuration, the situation is symmetric.
    This makes it possible to transfer the values of counters to the neighbouring configurations.

    We define the formula $\varphi_{\exists\rightarrow}(x)$ stating that the element $x$ admits a $t$-successor linked via $G_1$: \[\varphi_{\exists\rightarrow}(x)\eqdef\exists y.~t(x,y)\wedge G_1(x,y).\]
Notice that if two elements of two distinct configurations are linked via $t$, then they are in particular linked via precisely one of $E_2$ and $F_2$, which does not leave any ambiguity whether they are linked via $E_1$ or $F_1$.

	We also define the symmetric formula $\varphi_{\exists\leftarrow}(x)$ stating that $x$ admits a $t$-predecessor linked via $G_1$.
	
	We define now, for each $i\in\{1,2\}$, a formula $\varphi_{\forall i\exists\rightarrow}(x)$ stating that every element $y$ of the $G_2$-class of $x$ which satisfies $c_i$ admits a $t$-successor via $G_1$: \[\varphi_{\forall i\exists\rightarrow}(x)\eqdef\forall y.~\big(G_2(x,y)\wedge c_i(y)\big)\rightarrow\varphi_{\exists\rightarrow}(y).\]
	
    Because $E_1$ and $F_1$ are equivalence relations, we can see that if $\varphi_{\forall i\exists\rightarrow}(x)$ is true, then the configuration succeeding that of $x$ has a value of $c_i$ greater than or equal to the value of $c_i$ in the configuration of $x$. In order to make sure that it cannot increase by more than one, we add a condition stating that each two elements of the same configuration that satisfy $c_i$ and have no predecessor via $G_1$ are necessarily related via $E_1$, and hence at most one equivalence relation is added:
	\begin{align}
      \bigwedge_{i=1}^{2}\forall x,y.~&\big(G_2(x,y) \wedge c_i(x) \wedge c_i(y) \nonumber\\
		&\wedge\neg\phi_{\exists\leftarrow}(x) \wedge \neg\phi_{\exists\leftarrow}(y)\big) 
	~\rightarrow E_1(x,y)
	\end{align}

	Symmetrically, we define $\varphi_{\forall i\exists\leftarrow}(x)$ stating that every element $y$ of the $G_2$-class of $x$ which satisfies $c_i$ admits a $t$-predecessor via $G_1$, and we also define the symmetric condition making sure that each counter decreases by at most one at every step:
	\begin{align}
      \bigwedge_{i=1}^{2}\forall x,y.~&\big(G_2(x,y) \wedge c_i(x) \wedge c_i(y) \nonumber\\
		&\wedge\neg\phi_{\exists\rightarrow}(x) \wedge \neg\phi_{\exists\rightarrow}(y)\big) 
		~\rightarrow E_1(x,y)
	\end{align}

    We define below four auxiliary formulas $\varphi_{i,\textrm{zero}}(x)$, $\varphi_{i,\textrm{eq}}(x)$, $\varphi_{i,\textrm{incr}}(x)$ and $\varphi_{i,\textrm{decr}}(x)$.

	First, the formula $\varphi_{i,\textrm{zero}}(x)$ states that no element of the $G_2$-class of $x$ satisfy $c_i$, and hence the counter $c_i$ is zero:
    \[\varphi_{i,\textrm{zero}}(x)\eqdef\forall y.~G_2(x,y)\rightarrow\neg c_i(y).\]

    Then, the formula $\varphi_{i,\textrm{eq}}(x)$ states that the counter $c_i$ stays the same from the configuration of $x$ to the next one:
    \[\varphi_{i,\textrm{eq}}(x)\eqdef \varphi_{\forall i\exists\rightarrow}(x)\wedge\forall y.~t(x,y)\rightarrow \varphi_{\forall i\exists\leftarrow}(y).\]

	Finally, the formulas $\varphi_{i,\textrm{incr}}(x)$ and $\varphi_{i,\textrm{decr}}(x)$ state that the counter $c_i$, respectively, increments and decrements from the configuration of $x$ to the successive one:
    \begin{align*}
      \varphi_{i,\textrm{incr}}(x)\eqdef \varphi_{\forall i\exists\rightarrow}(x)\wedge\forall y.~t(x,y)\rightarrow \neg \varphi_{\forall i\exists\leftarrow}(y) \\
      \varphi_{i,\textrm{decr}}(x)\eqdef \neg \varphi_{\forall i\exists\rightarrow}(x)\wedge\forall y.~t(x,y)\rightarrow \varphi_{\forall i\exists\leftarrow}(y)
    \end{align*}

    Having the above formulas, as well as the unary state predicates, we can easily axiomatise all the transitions of the two-counter machine $\mathbf{M}$ together with its initial and final configurations; we leave details to the reader.

    Let $\Theta_{\mathbf{M}}$ be a sentence expressing the existence of a run of $\mathbf{M}$ that reaches the final state $s_F$. It is immediate to see that this sentence $\Theta_{\mathbf{M}}$ is finitely satisfiable if and only if ${\mathbf{M}}$ terminates. This allows us to conclude the undecidability of finite satisfiability.
    For general (infinite) satisfiability, we consider a sentence $\Theta'_{\mathbf{M}}$ expressing the existence of an infinite run of $\mathbf{M}$ (i.e. never reaching the final state $s_F$). Hence, the general satisfiability problem is also undecidable.

\section{Conclusion}\label{s:conclusion}

We analysed four extensions of \FOt{} that enable decidable reasoning over nested equivalence relations: \FOt$[\EQ]$, \FOt$[<,\EQ]$, \FOt$[\precEQ]$, and \FOt$[\succEQ]$, establishing the precise complexity of the finite satisfiability problems and bounds on the size of minimal finite models.
Except \FOt$[\EQ]$ that enjoys the finite model property, understanding general (infinite) satisfiability remains open.
Finally, we demonstrated that extending \FOt{} with two independent families of nested equivalence relations leads to undecidability.

\section*{Acknowledgments}
The first and second authors were supported by the Polish National Science Center grant No.~2021/41/B/ST6/00996. The third author was supported by the ERC grant INFSYS, agreement No.~950398.
We thank the first anonymous reviewer for his helpful feedback.

\bibliographystyle{kr}
\bibliography{mybib}  

@Article{GKV97,
   author="E. Gr{\"a}del and P. Kolaitis and M. Y. Vardi",
   title={On the decision problem for two-variable first-order logic},
   journal={Bulletin of Symbolic Logic},
   volume={3},
   number={1},
   year={1997},
   pages={53--69}
}

@Article{KO12,
  author    = {E. Kiero\'{n}ski and
               M. Otto},
  title     = {Small Substructures and Decidability Issues for First-Order
               Logic with Two Variables.},
   journal={Journal of Symbolic Logic},
   volume={77},
   year={2012},
   pages    = {729-765}
}

@Article{Ott01,
 author  = {Martin Otto},
 title   = {Two-Variable First-Order Logic over Ordered Domains},
 journal = {Journal of Symbolic Logic},
 volume  = {66},
 year    = {2001},
 pages   = {685-702}
}

@Article{Sco62,
   author={Dana Scott},
   title={A decision method for validity of sentences  in two variables},
   journal={Journal Symbolic Logic},
   volume={27},
   year={1962},
   pages={477}
}

@ARTICLE{KMP-HT14,
  author = {E. Kiero\'{n}ski and J. Michaliszyn and I. Pratt-Hartmann and
               L. Tendera},
    title = {Two-variable first-order logic with equivalence closure},
    journal = {SIAM Journal of Computing},
    year = {2014},
    volume = {43},
   number = {3},
    pages = {1012--1063}

}

@article{BMS09,
  author    = {M. Boja\'nczyk and
               A. Muscholl and
               T. Schwentick and
               L. Segoufin},
  title     = {Two-variable logic on data trees and XML reasoning},
  journal   = {J. ACM},
  volume    = {56},
  number    = {3},
  year      = {2009},
  ee        = {http://doi.acm.org/10.1145/1516512.1516515},
  bibsource = {DBLP, http://dblp.uni-trier.de}
}

@article{BDM11,
  author    = {M. Boja{\'{n}}czyk and
               C. David and
               A. Muscholl and
               T. Schwentick and
               L. Segoufin},
  title     = {Two-variable logic on data words},
  journal   = {{ACM} Trans. Comput. Log.},
  volume    = {12},
  number    = {4},
  pages     = {27},
  year      = {2011},
  }

@article{CW16b,
  author       = {Witold Charatonik and
                  Piotr Witkowski},
  title        = {Two-Variable Logic with Counting and Trees},
  journal      = {{ACM} Trans. Comput. Log.},
  volume       = {17},
  number       = {4},
  pages        = {31},
  year         = {2016},
  }

@article{KMW62,
  title={Entscheidungsproblem reduced to the $\forall \exists \forall$ case},
  author={Kahr, A.S. and Moore, E.F. and Wang, H.},
  year= {1962},
  journal= {Proc. Nat. Acad. Sci. U.S.A.},
	volume= {48},
	pages= {365-377},
	
	
}

@inproceedings{PH14,
 author = {I. Pratt-Hartmann},
 title = {Logics with Counting and Equivalence},
 booktitle = {Proceedings of the Joint Meeting of CSL and LICS},
 series = {CSL-LICS '14},
 year = {2014},
 pages = {76:1--76:10},
 articleno = {76},
 publisher = {ACM},
}

@article{StC13,
  author    = {Balder ten Cate and
               Luc Segoufin},
  title     = {Unary negation},
  journal   = {Logical Methods in Comp. Sc.},
  volume    = {9},
  number    = {3},
  year      = {2013},
}

@article{SZ12,
  author    = {Thomas Schwentick and
               Thomas Zeume},
  title     = {Two-Variable Logic with Two Order Relations},
  journal   = {Logical Methods in Computer Science},
  volume    = {8},
  number    = {1},
  year      = {2012}
	}

@inproceedings{DK19,
  author    = {Daniel Danielski and
               Emanuel Kiero\'{n}ski},
    title     = {Finite Satisfiability of Unary Negation Fragment with Transitivity},
  booktitle = {44th International Symposium on Mathematical Foundations of Computer
               Science, {MFCS}},
  series    = {LIPIcs},
  volume    = {138},
  pages     = {17:1--17:15},
   year      = {2019},
  url       = {https://doi.org/10.4230/LIPIcs.MFCS.2019.17},
  doi       = {10.4230/LIPIcs.MFCS.2019.17},
  bibsource = {dblp computer science bibliography, https://dblp.org}
}

@article{BBC16,
  author    = {Saguy Benaim and
               Michael Benedikt and
               Witold Charatonik and
               Emanuel Kiero\'{n}ski and
               Rastislav Lenhardt and
               Filip Mazowiecki and
               James Worrell},
  title     = {Complexity of Two-Variable Logic on Finite Trees},
  journal   = {{ACM} Trans. Comput. Log.},
  volume    = {17},
  number    = {4},
  pages     = {32:1--32:38},
  year      = {2016},
  url       = {http://dl.acm.org/citation.cfm?id=2996796},
  bibsource = {dblp computer science bibliography, https://dblp.org}
}

@inproceedings{Kie11,
  author    = {Emanuel Kiero\'{n}ski},
   title     = {Decidability Issues for Two-Variable Logics with Several Linear Orders},
  booktitle = {Computer Science Logic, 25th International Workshop / 20th Annual
               Conference of the EACSL, {CSL} 2011, Proceedings},
  series    = {LIPIcs},
  volume    = {12},
  pages     = {337--351},
   year      = {2011},
  url       = {https://doi.org/10.4230/LIPIcs.CSL.2011.337},
  doi       = {10.4230/LIPIcs.CSL.2011.337},
  bibsource = {dblp computer science bibliography, https://dblp.org}
}

@book{PH23,
title = "Fragments of First-Order Logic",
author = "Ian Pratt-Hartmann",
year = "2023",
series = "Oxford Logic Guides",
publisher = "Oxford University Press",
address = "United Kingdom",
}

@inproceedings{BB07,
  author = {Bj\"{o}rklund, Henrik and Boja\'nczyk, Mikolaj},
  title = {Shuffle expressions and words with nested data},
  year = {2007},
  booktitle = {Proceedings of the 32nd International Conference on Mathematical Foundations of Computer Science},
  pages = {750–761},
  series = {MFCS'07}
}

@inproceedings{HZ16,
author = {Zeume, Thomas and Harwath, Frederik},
title = {Order-Invariance of Two-Variable Logic is Decidable},
year = {2016},
url = {https://doi.org/10.1145/2933575.2933594},
doi = {10.1145/2933575.2933594},
booktitle = {Proceedings of the 31st Annual ACM/IEEE Symposium on Logic in Computer Science},
pages = {807–816},
series = {LICS '16}
}

@inproceedings{TZ20,
  author       = {Szymon Torunczyk and
                  Thomas Zeume},
  title        = {Register Automata with Extrema Constraints, and an Application to
                  Two-Variable Logic},
  booktitle    = {{LICS} '20: 35th Annual {ACM/IEEE} Symposium on Logic in Computer
                  Science, Saarbr{\"{u}}cken, Germany, July 8-11, 2020},
  pages        = {873--885},
  publisher    = {{ACM}},
  year         = {2020},
  doi          = {10.1145/3373718.3394748},
}

\appendix

\section{Additional Details on Theorems~\ref{t:three} and~\ref{t:one}}\label{s:BC}

In Section~\ref{s:A}, we established Theorem~\ref{t:two}, and in Section~\ref{s:sketch} we outlined the modifications required to adapt its proof to Theorems~\ref{t:three} and~\ref{t:one}.  
We provide now additional details on these modifications, supplementing the earlier sketch while avoiding a full repetition of the original argument.

Let us briefly recall the general strategy used in the proof of Theorem~\ref{t:two}.
Given a finitely satisfiable sentence~$\phi$, our goal is to show that $\phi$ admits a model of exponentially bounded size.
Let $\str{A} \models \phi$.
The central ingredient is Lemma~\ref{l:replacement}, which allows us to take a single equivalence class $\cC$ of~$\str{A}$ and replace it by a smaller subset $\cD \subseteq \cC$ so that the resulting structure with domain $(A \setminus \cC) \cup \cD$ is still a model of~$\phi$.
Applying Lemma~\ref{l:replacement} successively to each $E_k$-class, we shrink them one by one; a straightforward size analysis then yields a finite model of $\phi$ whose size is exponentially bounded.
In what follows, we explain how this method can be adapted to obtain Theorems~\ref{t:three} and~\ref{t:one}.

\smallskip
\noindent
\emph{Theorem~\ref{t:three}.}
We prove now Theorem~\ref{t:three}, stating that finitely satisfiable sentences of \FOt$[\precEQ]$ admit models of exponential size.
We proceed via a reduction to \FOt$[<,\EQ]$ equipped with an additional semantic constraint.

We observe that \FOt$[\precEQ]$-models correspond precisely to those \FOt$[<,\EQ]$-models that satisfy the following \emph{interval property}: for every $k \in \N$ and all $a,b,c \in A$ with $a <^{\str{A}} b <^{\str{A}} c$ and $(a,c) \in E^{\str{A}}_k$, it must hold that $b$ is $E^{\str{A}}_k$-related to both $a$ and $c$. Equivalently, each $E_k$-class forms an interval with respect to $<$.

The correspondence is explicit: given a structure $\str{A}$ interpreting $\preceq_1,\preceq_2,\dots$, choose any linear order $<^{\str{A}}$ consistent with $\preceq_1^{\str{A}}$ (i.e., $a \prec_1^{\str{A}} b \rightarrow a <^{\str{A}} b$).  
For each $k$, define $E^{\str{A}}_k$ as the equivalence relation induced by $\preceq_k^{\str{A}}$ (namely, $a E_k b$ iff both $a \preceq_k b$ and $b \preceq_k a$). The resulting expansion of $\str{A}$ satisfies the interval property by construction.  

In the class of models with the interval property, the total preorder symbol \(\preceq_k\) is axiomatisable in \FOt$[<,\EQ]$ using the formula \(x{\le}y \vee x E_k y\). We can thus rewrite any formula in \FOt$[\precEQ]$ to equivalent one in \FOt$[<,\EQ]$.

Conversely, any \FOt$[<,\EQ]$-structure with the interval property naturally induces a family of nested preorders, yielding an \FOt$[\precEQ]$-model.

To complete the reduction, it remains to ensure that the interval property
is preserved by the model transformation in Lemma~\ref{l:replacement}.
This is captured in the following claim.

\begin{claim}\label{cl:interval-preservation}
Let $\str{A}$ be an \FOt$[<,\EQ]$-structure satisfying the interval property, and let $\str{B}$ be obtained from $\str{A}$ via the construction in Lemma~\ref{l:replacement}.
Then $\str{B}$ also satisfies the interval property.
\end{claim}
\begin{proof}
In the construction of Lemma~\ref{l:replacement}, the new domain is 
$B = (A \setminus \cC) \cup \cD$, where $\cD \subseteq \cC$ consists of whole $E_k$-classes.
Since the interval property is preserved under taking substructures, the reduct $\str{A} \restr B$ inherits it from $\str{A}$.
Moreover, by Item~\ref{smallclaimfo2item2} of Claim~\ref{l:smallclaimfo2}, 
the interpretations of $<$ and of all $E_k$ in $\str{B}$ coincide with those in $\str{A} \restr B$.
Hence $\str{B}$ also has the interval property.
\end{proof}

This shows that the small model construction of Theorem~\ref{t:two} preserves the intended subclass of models with the interval property, thereby completing the proof of Theorem~\ref{t:three}.

\smallskip
\noindent
\emph{Theorem~\ref{t:one}.}
We now establish Theorem~\ref{t:one}, asserting that $\FOt[\EQ]$ enjoys the exponential model property: every satisfiable sentence admits a finite model of size bounded exponentially in the length of the sentence.
Our proof follows again the strategy used for $\FOt[<,\EQ]$ in Lemma~\ref{l:replacement} and Theorem~\ref{t:two}, with the key distinction that we no longer assume the initial model $\str{A}$ to be finite.

In Lemma~\ref{l:replacement}, finiteness of $\str{A}$ was used only in defining $r$-extremal subsets.
Recall that for a finite set $S$, ${\rm extr}_r(S)$ consists of the $r$ minimal and $r$ maximal elements of $S$ with respect to the linear order $<$, or is equal to $S$ if $|S|<2r$.
Without~$<$, we redefine ${\rm extr}_r(S)$ as $S$ if $|S|<2r$, and as any subset of $S$ of size $2r$ otherwise.
This modification removes any dependence on finiteness or ordering.

The remainder of the construction from Lemma~\ref{l:replacement} carries over almost verbatim to $\FOt[\EQ]$.
We proceed as in the original proof by defining the set $W_{\RN{1}}$, and analogously $W_{\RN{2}}$ and $W_{\RN{3}}$.
The replacement set $\cD$ for $\cC$ is then defined, as before, as the $E_k$-closure of $W_{\RN{1}} \cap W_{\RN{2}} \cap W_{\RN{3}}$ restricted to~$\cC$.
Claims~\ref{l:configurations} and~\ref{l:sizeanalysis} remain valid, with the only adjustment that in Item~\ref{configurationitems2} of the definition of $a$-configurations we now consider equality alone, omitting any reference to orders.

The only apparent obstacle is in the final part of the proof. As starting from an infinite sequence $a_1,a_2,\dots$ of elements of \(B \setminus (W_{\RN{1}} \cup W_{\RN{2}})\), constructing the sequence of structures
\(\str{B}_0, \str{B}_1, \str{B}_2, \dots\)
could in principle require infinitely many steps.
Recall that, for every $n\in\N$, $\str{B}_{n+1}$ is obtained from $\str{B}_n$ by modifying some $2$-types between $a_{n+1}$ and elements $b \in W_{\RN{1}} \setminus \absclass{\str{A}}{k}{a_{n+1}}$.
However, for any fixed $b$, the $2$-type $\tp^{\str{A}}[a_{n+1},b]$ is altered only at the step of the construction corresponding to $a_{n+1}$, and no later step revisits it.
Thus we may define the final structure $\str{B}$ as the natural limit structure of that sequence $\str{B}_0,\str{B}_1,\dots$: $\tp^{\str{B}}[a_{n+1},b]$ is set to the \emph{ultimate} $2$-type for that pair.
We conclude that Claim~\ref{l:smallclaimfo2} also remains valid, and therefore Lemma~\ref{l:replacement}.

Since the size bound analysis from the proof of Theorem~\ref{t:two} still applies, the small model is finite and of exponentially bounded size, completing the proof of Theorem~\ref{t:one}.

\section{Proof of Lemma~\ref{l:pumping}}\label{appendix:D}

In this appendix, we prove Lemma~\ref{l:pumping}.
The proof builds on the approach used in Lemma~\ref{l:replacement}, but its technical details are considerably more involved.

\pumpinglemma*

For this whole appendix, we fix a finitely satisfiable sentence $\phi$ of $\FOt[\succEQ]$.
As explained in Subsection~\ref{s:snf}, we may assume that $\phi$ is in Scott Normal Form, i.e.,
in the shape of \(\forall x,y.~\psi_0(x,y) \wedge \bigwedge_{m=1}^{M}{\forall x.~\exists y.~\psi_m(x,y)} \),
where $M$ denotes the number of Skolem conjuncts.
By convention, $\phi$ uses the first $K$ nested total preorder symbols $\preceq_1,\dots,\preceq_K$ and the corresponding derived symbols $\cS_k$ and $E_k$,
where $\cS_k$ is the induced successor relation of $\preceq_k$, and $E_k$ is the equivalence relation induced by $\preceq_k$-equivalent elements.
Interpretations of $E_1,\dots,E_K$ are therefore naturally nested.
We also include two auxiliary symbols $E_0$ and $E_{K+1}$ interpreted as the identity and universal relations, respectively.
Let $\AAA$ denote the set of all $1$-types over the signature of $\phi$.

Suppose that $\str{A}$ is a finite model of $\phi$, $k$ is a number between $0$ and $K$, and $\cC$ is an $E_{k+1}$-class from $A/E^{\str{A}}_{k+1}$.

We assume w.l.o.g.~that the domain $A$ is a finite subset of~$\N$ and that the total preorders $\preceq_k^{\str{A}}$ are consistent with the natural order $\le$ on $\N$ (i.e., $a \prec_k b \rightarrow a < b$).

We fix \emph{witness functions} $\ff_m\colon A \rightarrow A$, i.e., functions such that for each $m \in [M]$ and every $a \in A$, $\str{A}\models \psi_m(a,\ff_m(a))$.

Let us recall witness-related terminology of Section~\ref{s:A}:
if $m \in [M]$ and $a,b$ are elements of $A$, then we call $b$ a \emph{witness} of $a$ whenever $\str{A} \models \psi_m(a,b)$.
In addition, $b$ is an \emph{internal} witness if $(a,b) \in E^{\str{A}}_k$; an \emph{external} witness if $b \in A \setminus \cC$ and $(a,b) \not\in E^{\str{A}}_k$; and a \emph{$\cC$-witness} if $b \in \cC$ and $(a,b) \not \in E^{\str{A}}_k$.

We need also to introduce some new notions, addressing the presence of the induced successors $\cS_1,\cS_2,\dots$.

A \emph{$\cC$-neighbourhood} of an element $a \in A$ is the set
\[ \cN_{\cC}(a) := \big\{ b \in \cC ~\big|~ (a,b) \in E^{\str{A}}_k \cup \cS^{\str{A}}_k \text{ or } (b,a) \in \cS^{\str{A}}_k \big\}. \]
Let $b \in A$ be a $\cC$-witness of $a$. Then $b$ is a \emph{local} witness of $a$ if $b \in \cN_{\cC}(a)$; and otherwise, $b$ is a \emph{remote} witness.

We adopt the notion of \emph{configurations} introduced in Section~\ref{s:A}.
Let $\ell \in [M]$ and let $a \in A$. Assume that $\langle b_1, \dots, b_\ell \rangle$ and $\langle c_1, \dots, c_\ell \rangle$ are tuples of pairwise distinct elements of $A$.
We say that they realise the same \emph{$a$-configuration} if the following conditions hold:
\begin{enumerate}[label = (\roman*)] \itemsep0pt
  \item For every $i \in [\ell]$, $\type{\str{A}}{b_i} = \type{\str{A}}{c_i}$.
  \item For every $i \in [\ell]$, $a \bowtie b_i$ if and only if $a \bowtie c_i$, where $\bowtie$ is any of the symbols ``$<$'', ``$=$'' or ``$>$''.
  \item For every $i \in [\ell]$ and every $j \in [K]$, $(a,b_i)\in R^{\str{A}}_{j}$ if and only if $(a,c_i)\in R^{\str{A}}_{j}$, where $R_j$ is any of the symbols ``$E_j$'', ``$\cS_j$'' or ``$\cS^{-1}_j$''. Here, the last symbol $\cS^{-1}_j$ is the \emph{inverse} of $\cS_j$, i.e. the induced predecessor of $\preceq_j$.
\end{enumerate}

Finally, we will also use the notion of the \emph{$r$-extremal} subset of a finite subset $S\subseteq\N$, of $\ell$ elements $a_1< a_2<\ldots< a_{\ell}$: if $\ell\leq2\cdot r$, then this subset, denoted by ${\rm extr}_r(S)$, is $S$ itself; if $\ell> 2\cdot r$, then it is $\{ a_1,\dots,a_{r},a_{\ell-r+1},\dots,a_\ell \}\subset S$.

\smallskip
\noindent
{\bf Securing witnesses.}
As in the proof of Lemma~\ref{l:replacement}, we aim to select three small sets $W_{\RN{1}}, W_{\RN{2}}, W_{\RN{3}}$ providing witnesses for each other in a circular fashion.
However, the difference is that we now treat 
local and remote witnesses separately. 
For remote witnesses, we work with three \emph{global} sets, while for local witnesses 
we will select three sets \emph{within each} adjacent pair of $E_k$-classes inside~$\cC$.

\smallskip\noindent
\emph{Remote witnesses.}
We define 
\[
W^{\mathrm{rem}}_{\RN{1}} \ \eqdef\ \bigcup_{\alpha \in \AAA} \mathrm{extr}_{M}(\{ a \in \cC \mid \type{\str{A}}{a} = \alpha \}), 
\]
then
\[
W^{\mathrm{rem}}_{\RN{2}} \ \eqdef\ \bigcup_{m \in [M]} \ff_m\!\left[ W^{\mathrm{rem}}_{\RN{1}} \right],
\]
and finally
\[
W^{\mathrm{rem}}_{\RN{3}} \ \eqdef\ \bigcup_{m \in [M]} \ff_m\!\left[ W^{\mathrm{rem}}_{\RN{2}} \right].
\]
Set
\[
W^{\mathrm{rem}} \ \eqdef\ \left( W^{\mathrm{rem}}_{\RN{1}} \ \cup\ 
W^{\mathrm{rem}}_{\RN{2}} \ \cup\ W^{\mathrm{rem}}_{\RN{3}} \right) \cap \cC.
\]

In the next two claims we complete the analysis for remote witnesses: 
Claim~\ref{l:configurations2} shows that any remote $\cC$-witness configuration can be realised 
within $W^{\mathrm{rem}}_{\RN{1}}$, while Claim~\ref{l:siemaclaim1} bounds the total size of $W^{\mathrm{rem}}$.

\begin{claim}\label{l:configurations2}
Let $\ell \in [M]$ and let $a \in A \setminus W^{\mathrm{rem}}_{\RN{1}}$. 
Suppose that $b_1,\dots,b_\ell$ are distinct elements of 
$\cC \setminus \cN_{\cC}(a)$. 
Then there exist distinct elements $c_1,\dots,c_\ell$ in $W^{\mathrm{rem}}_{\RN{1}}$ 
such that the tuples $\langle b_1,\dots,b_\ell \rangle$ 
and $\langle c_1,\dots,c_\ell \rangle$ realise the same $a$-configuration.
\end{claim}

\begin{claim}\label{l:siemaclaim1}
The size of $W^{\mathrm{rem}}$ is at most $6 \cdot M^3 \cdot |\AAA|$.
\end{claim}

\smallskip
\noindent
\emph{Local witnesses.}
We now turn our focus to local witnesses.
Our goal is to associate with each pair of adjacent $E_k$-classes within $\cC$ a small set of elements sufficient to realise any possible configuration of local witnesses.
A subtlety arises because an element may require witnesses that are connected via $\cS_{j+1}, \cS_{j+2}, \dots$ but not via $\cS_{1}, \dots, \cS_{j}$, for some $j \in [k]$.
To handle this, we prepare the sets in a stratified manner, reflecting the levels $1$ through $k$ of the nested equivalence relations.

Let $\cE_1 \prec^{\str{A}}_k \dots \prec^{\str{A}}_k \cE_L$ be the sequence of all $E_k$-classes in $\cC / E^{\str{A}}_k$, ordered according to $\prec^{\str{A}}_k$ (i.e., $\cE \prec^{\str{A}}_k \cE'$ iff $a \prec^{\str{A}}_k b$ for each $a \in \cE$ and each $b \in \cE'$).  
Here, $L$ denotes the total number of such $E_k$-classes, i.e., $L = |\cC/E^{\str{A}}_k|$.  

For each $i \in [L]$ and each level $j \in [k{-}1]$, we denote by $\cE_{i,j}^{\min}$ and $\cE_{i,j}^{\max}$, respectively, the $\preceq^{\str{A}}_j$-minimal and $\preceq^{\str{A}}_j$-maximal $E_j$-classes within the quotient $\cE_i / E^{\str{A}}_j$.  
For convenience, we introduce sentinel sets $\cE^{\min}_{i,0} \eqdef \emptyset$ and $\cE^{\max}_{i,0} \eqdef \emptyset$, corresponding to the trivial level $j=0$.

For every $i \in [L]$, we choose $V_i \subseteq \cE_i$ as
\[\bigcup_{\alpha \in \AAA}\bigcup_{j=1}^{k}{{\rm extr}_M\big(\{ a \in \cE_i \setminus (\cE^{\rm min}_{i,j-1} \cup \cE^{\rm max}_{i,j-1}) \mid \type{\str{A}}{a} = \alpha \}\big)}.\]
That is, for each $1$-type $\alpha$ and each level $1 \leq j \leq k$, we take the $M$-extremal elements of those members of $\cE_i$ of type $\alpha$ that are outside the $(j{-}1)$-level minimal and maximal $E_{j-1}$-classes $\cE_{i,j-1}^{\min}$ and $\cE_{i,j-1}^{\max}$.

By construction, each set $V_i$ is a small subset of $\cE_i$ that will suffice to realise every local witness configuration for elements that lie inside the two adjacent $E_k$-classes.

\begin{claim}\label{l:configurations3}
Let $\ell \in [M]$, let $j\in[L]$, and let $a\in\cE_j$.  
Suppose $b_1,\dots,b_\ell$ are distinct elements of some neighbouring class $\cE_i$ with $|i-j|=1$.  
Then there exist distinct elements $c_1,\dots,c_\ell\in V_i$ such that the tuples $\langle b_1,\dots,b_\ell\rangle$ and $\langle c_1,\dots,c_\ell\rangle$ realise the same $a$-configuration.
\end{claim}

The following claim is an easy corollary of Claim~\ref{l:configurations3}:
\begin{claim}\label{l:elosameonetypes}
  For every $i \in [L]$, we have that
  \[ \{ \type{\str{A}}{a} \mid a \in V_i \} = \{ \type{\str{A}}{a} \mid a \in \cE_i \}. \]
\end{claim}

We now define a local witnessing set for each adjacent pair of $E_k$-classes $\cE_i,\cE_{i+1}$, where $i\in[L-1]$.  
First set
\[
W^{\mathrm{loc}}_{\RN{1}}[i,i{+}1] \eqdef V_i \cup V_{i+1}.
\]
Then define
\[
W^{\mathrm{loc}}_{\RN{2}}[i,i{+}1] \eqdef \bigcup_{m\in[M]} \ff_m\big(W^{\mathrm{loc}}_{\RN{1}}[i,i{+}1]\big),
\]
and
\[
W^{\mathrm{loc}}_{\RN{3}}[i,i{+}1] \eqdef \bigcup_{m\in[M]} \ff_m\big(W^{\mathrm{loc}}_{\RN{2}}[i,i{+}1]\big).
\]
Finally choose $W^{\mathrm{loc}}[i,i{+}1]$ as
\[
\bigl(W^{\mathrm{loc}}_{\RN{1}}[i,i{+}1] \cup 
W^{\mathrm{loc}}_{\RN{2}}[i,i{+}1] \cup W^{\mathrm{loc}}_{\RN{3}}[i,i{+}1]\bigr)
\cap (\cE_i\cup\cE_{i+1}).
\]

\begin{claim}\label{l:siemaclaim2}
For each $i\in[L-1]$, the set $W^{\mathrm{loc}}[i,i{+}1]$ has size at most \(12\cdot M^{3}\cdot K\cdot |\AAA|\).
\end{claim}

\smallskip
\noindent
\textbf{Reducible pairs.}
Having secured both local and remote $\cC$-witnesses in the designated sets, 
we now show that if $L$---the number of $E_k$-classes within $\cC$---exceeds a certain threshold, then certain $E_k$-classes of $\cC$ can be removed without affecting the satisfaction of~$\phi$.

Let $p$ and $q$ be integers with $1 < p \le q < L$.  
We say that an isomorphism between substructures
\[
\rho \colon \str{A}~\restr~W^{\mathrm{loc}}[p{-}1,p] \longrightarrow \str{A}~\restr~W^{\mathrm{loc}}[q,q{+}1]
\]
is \emph{canonical} if:
\begin{enumerate}\itemsep0pt
  \item $\rho$ maps $V_{p{-}1}$ to $V_{q}$ and $V_{p}$ to $V_{q{+}1}$;
  \item for each $m \in [M]$ and every $a, b \in W^{\mathrm{loc}}[p{-}1,p]$,  
  we have that $b = f_m(a)$ iff $\rho(b) = f_m(\rho(a))$.
\end{enumerate}

This setup of canonical isomorphisms is tailored for the subsequent reduction step:  
If \(W^{\mathrm{loc}}[p{-}1,p]\) and \(W^{\mathrm{loc}}[q,q{+}1]\) are canonically isomorphic,  
then, after removing all \(E_k\)-classes \(\cE_p, \dots, \cE_q\),  
the classes \(\cE_{p-1}\) and \(\cE_{q+1}\) become consecutive.  
In this case, the canonical isomorphism allows us to naturally \emph{glue} together  
the local witnessing configurations of \([p{-}1,p]\) and \([q,q{+}1]\)  
to obtain a new local witnessing block \([p{-}1,q{+}1]\).  

With the above intuitions in mind, 
we introduce the notion of reducible pairs pinpointing
segments of consecutive $E_k$-classes that may be safely removed
while still preserving the satisfaction of~$\phi$.
Formally, a pair \((p,q)\) of integers with $1 < p \le q < L$ is \emph{reducible} if:
\begin{itemize}\itemsep0pt
  \item\label{defreduciblepair1} For every $i \in [p{-}1,\,q{+}1]$, we have $\cE_i \cap W^{\mathrm{rem}} = \emptyset$.
  \item\label{defreduciblepair2} 
    The substructures of $\str{A}$ induced by the sets $W^{\mathrm{loc}}[p{-}1,p]$ and $W^{\mathrm{loc}}[q,q{+}1]$ are canonically isomorphic.
\end{itemize}

\begin{lemma}\label{l:reducespan}
  Let $\str{A}$ be a finite model of $\phi$. Fix $k \in [0,K]$,  
  and let $\cC \in A / E^{\str{A}}_{k+1}$.
  Suppose that $\cE_1 \prec^{\str{A}}_k \dots \prec^{\str{A}}_k \cE_L$ are precisely the $E_k$-classes of $\cC / E^{\str{A}}_k$.
  Assume that the sets $W^{\mathrm{rem}}$ and
  $W^{\mathrm{loc}}[1,2],\dots,W^{\mathrm{loc}}[L{-}1,L]$
  have been defined as above.
  If $(p,q)$ is a reducible pair for some $1 < p \le q < L$,  
  then there exists a model $\str{B}$ of $\phi$ with domain
  \[
    B \;=\; A \setminus \bigcup_{i \in [p,q]} \cE_i.
  \]
\end{lemma}
\begin{proof}[Proof]
Let initially $\str{B}$ denote the structure $\str{A}\restr B$.
In this structure $\str{B}$, $\cE_{p-1}$ and $\cE_{q+1}$ are now adjacent $E_k$-classes.
Since $(p,q)$ is assumed to form a reducible pair, the sets $W^{\rm loc}[p{-}1,p]$ and $W^{\rm loc}[q,q{+}1]$ are canonically isomorphic. 
Therefore, it makes sense to speak about $W^{\rm loc}[p{-}1,q{+}1]$, defined in a natural way:
\[
\big(W^{\rm loc}[p{-}1,p] \cap \cE_{p{-}1}\big) \cup \big(W^{\rm loc}[q,q{+}1] \cap \cE_{q{+}1}\big).
\]
We also adjust the structure $\str{B}$ so that $\str{B}~\restr~W^{\rm loc}[p{-}1,q{+}1]$ is isomorphic to $\str{A}~\restr~W^{\rm loc}[p{-}1,p]$ in an obvious way.

Now, we adjust some $2$-types in the so-modified structure $\str{B}$ so that it provides replacements for all witnesses originally provided in $\bigcup_{i \in [p,q]}\cE_{i}$.
This is similar to the final part of the proof of Lemma~\ref{l:replacement} (i.e., the part in which we construct a sequence of structures $\str{B}_0,\str{B}_1,\dots,\str{B}_N$).

For each element of $a \in B \setminus (W^{\rm rem}_{\RN{1}} \cup W^{\rm rem}_{\RN{2}})$,
we search for its remote witnesses in the set $W^{\rm rem}_{\RN{1}}$ using Claim~\ref{l:configurations2}.
Let $b_1,\dots,b_\ell$ be the sequence of all remote $\cC$-witnesses of $a$, where $\ell \le M$.
Using Claim~\ref{l:configurations2}, we find new elements $c_1,\dots,c_\ell \in W^{\rm rem}_{\RN{1}}$ realising the same $a$-configuration.
We set $\type{\str{B}}{a,c_i} \eqdef \type{\str{A}}{a,b_i}$ for every $i \in [\ell]$.

Having remote witnesses secured, we now search for the replacements of local witnesses for elements in $\cE_{p-1}\cup\cE_{q+1}$. 

Fix an element $a \in \cE_{p-1}$ whose local $\cC$-witnesses are possibly missing in $\str{B}$, i.e., $a$ is from the set
\[ \cE_{p-1} \setminus \big(W^{\rm loc}_{\RN{1}}[p{-}1,p] \cup W^{\rm loc}_{\RN{2}}[p{-}1,p]\big).\]
Let $b_1,\dots,b_\ell$ be the sequence of all local $\cC$-witness of $a$, where $\ell \le M$, that were originally contained in the set $\cE_{p}$.
That is, $b_1,\dots,b_\ell$ are the distinct elements of the set
\[ \big\{ \ff_m(a) \mid m \in [M] \big\} \cap \cE_{p}. \]
Using Claim~\ref{l:configurations3}, we find new elements $c_1,\dots,c_\ell \in V_{q+1}$ realising the same $a$-configuration as $b_1,\dots,b_\ell$.
We set $\type{\str{B}}{a,c_i} \eqdef \type{\str{A}}{a,b_i}$ for every $i \in [\ell]$.

Notice that local $\cC$-witnesses of $a$ that were originally in $\cE_{p-2}$ remain present also in $\str{B}$.

In a symmetric way, we provide local $\cC$-witnesses for elements in the set
\[ \cE_{q{+}1} \setminus \big(W^{\rm loc}_{\RN{1}}[q,q{+}1] \cup W^{\rm loc}_{\RN{2}}[q,q{+}1]\big).\]

Therefore, at this moment, every element of $\str{B}$ has all of its local and remote $\cC$-witnesses secured as well as internal and external witnesses.

Finally, we adjust the $2$-types between $\cE_{p-1}$ and $\cE_{q+1}$ that were not modified in any of the previous steps.
This is necessary, as these $2$-types were originally connecting elements that were not $\cS_k$-successive, yet now they are.

Fix $a \in \cE_{p-1}$ and $b \in \cE_{q+1}$ such that the $2$-type of $(a,b)$ is unmodified, i.e., as in $\str{A}$.
Using Claim~\ref{l:elosameonetypes} and the fact that $\str{A}~\restr~V_{p}$ and $\str{A}~\restr~V_{q+1}$ are isomorphic, we can find $b' \in \cE_{p}$ of the same $1$-type as $b$.
We set $\type{\str{B}}{a,b} \eqdef \type{\str{A}}{a,b'}$.

The construction of the model $\str{B}$ for $\phi$ is finished.
\end{proof}

We now show that long enough sequences of $E_k$-classes
must necessarily contain a reducible pair.
\begin{lemma}\label{l:thresholdspan}
  Let $\str{A}$ be a finite model of $\phi$, let $k \in [0,K]$, and let $\cC \in A/E^{\str{A}}_{k+1}$.
  Suppose that $\cE_1 \prec^{\str{A}}_k \dots \prec^{\str{A}}_k \cE_L$ are precisely the $E_k$-classes of $\cC / E^{\str{A}}_k$.
  Assume that the sets $W^{\mathrm{rem}}$ and $W^{\mathrm{loc}}[1,2],\dots,W^{\mathrm{loc}}[L{-}1,L]$ are defined as above.
  If
  \[
    L \;\ge\; c \cdot |W^{\mathrm{rem}}| \cdot 2^{n^2 \times (|\Rels|+M)},
  \]
  where $c >0$ is a constant independent from $\varphi$, and
	\[n = \max_{i \in [L{-}1]} \big|W^{\mathrm{loc}}[i,i{+}1]\big|\]
  and \(\Rels \text{ is the set of relation symbols of }\phi\),
  then there exists a reducible pair $(p,q)$ with $1 < p \le q < L$.
\end{lemma}
\begin{proof}
This lemma is essentially the ``pumping lemma'' for regular languages.

First we notice that $2^{n^2 \cdot |\Rels|}$ bounds the number of distinct structures over a domain of cardinality $n$;
recall that $\phi$ is in Scott Normal Form, and thus it does not use constants and mentions only relation symbols of arity $1$ and $2$.
Likewise, $2^{n^2 \cdot M}$ bounds the number of possible configurations of witnessing functions.

Therefore, if $L$ exceeds the threshold $2^{n^2 \times (|\Rels| + M)}$, necessary canonically isomorphic structures $\str{A} \restr W^{\rm loc}[i,i{+}1]$ exist for at least two distinct positions $i \in [L-1]$, say, $i \eqdef p{-}1$ and $i \eqdef q$. Assume that $1 < p \le q < L$.

Moreover, if this threshold is additionally multiplied by $c \cdot |W^{\rm rem}|$, for a large enough constant $c > 0$, we can choose $p$ and $q$ so that $\bigcup_{i\in[p-1,q+1]}{\cE_i}$ is disjoint from $W^{\rm rem}$.

Consequently, $(p,q)$ is a reducible pair.
\end{proof}

Having Lemmas~\ref{l:reducespan} and~\ref{l:thresholdspan}, the proof of Lemma~\ref{l:pumping} is nearly complete;  
it remains only to establish the upper bound on the set $\cD$ stated in its first item.

From Lemma~\ref{l:reducespan}, it follows that $\phi$ has a model $\str{B}$ with domain $(A \setminus \cC) \cup \cD$, in which the number of $E_k$-classes in the set $\cD \subseteq \cC$, denoted here by $L$, is below the threshold from Lemma~\ref{l:thresholdspan}.

Let the numbers $c$, $n$, and the set $\Rels$ be as in the statement of Lemma~\ref{l:thresholdspan}.  
Recall that $\AAA$ is the set of all $1$-types over the signature of $\phi$.
By Claim~\ref{l:siemaclaim2}, we have that:
\begin{align*}
  2^{n^2 \cdot (|\Rels|+M)} 
    &\le 2^{(12 \cdot M^3 \cdot K \cdot |\AAA|)^2 \cdot (|\Rels|+M)} \\
    &= 2^{\cO(M^7 \cdot K^2 \cdot |\AAA|^2 \cdot |\Rels|)}.
\end{align*}
Let us denote the exponent $\cO(M^7 \cdot K^2 \cdot |\AAA|^2 \cdot |\Rels|)$ by $d$.  
From Claim~\ref{l:siemaclaim1}, it follows that
\[
  |W^{\mathrm{rem}}| = 6 \cdot M^3 \cdot |\AAA| 
    = 2^{\cO(\log(M) + \log(|\AAA|))},
\]
which is asymptotically negligible compared to $2^d$.  
Consequently, we conclude with
\[
L \le c \cdot |W^{\mathrm{rem}}| \cdot 2^{d} = 2^{\cO(M^7 \cdot K^3 \cdot |\AAA|^2 \cdot |\Rels|)} = 2^{2^{\cO(|\phi|)}},
\]
as $M$, $K$ and $|\Rels|$ are all $\cO(|\phi|)$, and $|\AAA|$ is $2^{\cO(|\phi|)}$.

\section{Details on the \ExpSpace{} Lower Bound} \label{appendix:lower}

In Section~\ref{s:D}, we sketched the proof of the \ExpSpace{} lower bound for the finite satisfiability problem for \FOt$[\succEQ]$.  
Here, we provide some additional details.  
In fact, the result holds already for a highly restricted fragment of \FOt$[\succEQ]$:

\begin{proposition}
  The finite satisfiability problem is \ExpSpace{}-hard
  for the constant-free, equality-free, monadic fragment of \FOt{}
  extended with a single total preorder relation~$\preceq_1$ and its induced successor~$\cS_1$.
\end{proposition}

The reduction is from the \emph{corridor tiling problem}, which is known to be \ExpSpace{}-complete.  
An instance of this problem is a tuple $\cT = \langle \cC, c_{0}, c_{1}, \cH, \cV, n \rangle$, where $\cC$ is a non-empty finite set of \emph{colours},  
$c_{0} \in \cC$ is the \emph{initial} colour, $c_{1} \in \cC$ is the \emph{final} colour, $\cH, \cV \subseteq \cC \times \cC$ are the sets of \emph{horizontal} and \emph{vertical} constraints, and $n$ is a natural number given in unary, respectively.  
Then the instance $\cT$ is \emph{solvable} if there exists $m \in \N$ and a \emph{tiling} function $f\colon \cG \to \cC$ of a grid $\cG \eqdef [0,2^n{-}1] \times [0,m{-}1]$ such that:
\begin{itemize}\itemsep=0pt
  \item the initial and final conditions hold: $f(0,0) = c_0$ and $f(2^n{-}1,m{-}1) = c_1$;
  \item the horizontal constraints hold: for all $i \in [0,2^n{-}2], j \in [0,m{-}1]$, we have $\langle f(i,j), f(i{+}1,j) \rangle \in \cH$;
  \item the vertical constraints hold: for all $i \in [0,2^n{-}1], j \in [0,m{-}2]$, we have $\langle f(i,j), f(i,j{+}1) \rangle \in \cV$.
\end{itemize}
Finally, the problem is to decide if the instance is solvable.

\smallskip
In the following, given a corridor tiling instance $\cT$,  
we construct a formula $\Theta$ that is (finitely) satisfiable if and only if $\cT$ is solvable. 

We introduce $n$ unary predicates $B_j$ for $0 \le j < n$,  
together with a unary predicate $P_c$ for each colour $c \in \cC$.  

The predicates $B_j$ serve as bits, allowing us to associate with each element 
a number between $0$ and $2^{n}-1$ in the natural binary encoding.  
We write $\phi_{\mathrm{eq}}(x,y)$ and $\phi_{\mathrm{succ}}(x,y)$ for standard (polynomial-size in $n$) formulas asserting that 
the numbers assigned to $x$ and $y$ are equal, respectively successive.  
Similarly, $\phi_{\mathrm{all\text{-}off}}(x)$ and $\phi_{\mathrm{all\text{-}on}}(x)$ are abbreviations for formulas 
stating that all the $B_j$ are false (so $x$ encodes $0$), and, respectively,  all the $B_j$ are true 
(so $x$ encodes $2^n{-}1$).

The intended encoding is as follows.  
Domain elements correspond to individual grid cells.
The predicates $B_j$ determine the horizontal coordinate,  
while each $E_1$-equivalence class represents an entire row.  
Within a row, the predicate $P_c$ marks the colour assigned to the corresponding grid cell.  
The successor relation $\cS_1$ connects consecutive $E_1$-classes (i.e., $\preceq_1$-equivalent elements),  
thereby linking each row to its immediate vertical successor.  
In this way, the structure naturally represents a grid of width $2^n$ and height $m$.
The reduction is enforced by a collection of formulas ensuring that the structure encodes a valid tiling.  

\smallskip\noindent
\emph{Grid formation.}  
Every row must contain the leftmost element (position~$0$), and all elements except the last must have a successor,  
ensuring that each row forms a complete horizontal sequence of width $2^n$:

\[
  \forall x.~\exists y.~\phi_{\mathrm{all\text{-}off}}(y) \wedge x E_1 y
\]

\[
  \forall x.~\neg \phi_{\mathrm{all\text{-}on}}(x) \rightarrow \exists y.~\phi_{\mathrm{succ}}(x,y) \wedge x E_1 y
\]

\smallskip\noindent
\emph{Boundary conditions.}  
The leftmost cell of the first row is coloured with $c_0$, and  
the rightmost cell of the last row is coloured with $c_1$:

\[
  \forall x.~\big(\phi_{\mathrm{all\text{-}off}}(x) \wedge \neg \exists y.~\cS_1(y,x)\big) \rightarrow P_{c_0}(x)
\]

\[
  \forall x.~\big(\phi_{\mathrm{all\text{-}on}}(x) \wedge \neg \exists y.~\cS_1(x,y)\big) \rightarrow P_{c_1}(x)
\]

\smallskip\noindent
\emph{Horizontal consistency.}  
Within a row, horizontally aligned positions must agree on their colour,  
and successive horizontal positions must satisfy the horizontal constraints~$\cH$:

\[
  \forall x,y.~\big(x E_1 y \wedge \phi_{\mathrm{eq}}(x,y)\big) \rightarrow \bigwedge_{c \in \cC} \big(P_c(x) \leftrightarrow P_c(y)\big)
\]

\[
  \forall x,y.~\big(x E_1 y \wedge \phi_{\mathrm{succ}}(x,y)\big) \rightarrow \bigvee_{(c,c') \in \cH} \big(P_c(x) \land P_{c'}(y)\big)
\]

\smallskip\noindent
\emph{Vertical consistency.}  
Finally, consecutive rows are linked via $\cS_1$,  
and corresponding positions must respect the vertical constraints~$\cV$:

\[
  \forall x,y.~\big(\cS_1(x,y) \wedge \phi_{\mathrm{eq}}(x,y)\big) \rightarrow \bigvee_{(c,c') \in \cV} \big(P_c(x) \land P_{c'}(y)\big)
\]

Together, these formulas force any finite model of $\Theta$ to encode exactly a valid tiling of a $2^n \times m$ grid, for some $m \in \N$,
satisfying both boundary and propagation conditions.  
This completes the reduction.

\end{document}